\documentclass[preprint,1p]{elsarticlearxiv}

\usepackage[dvipsnames,usenames]{color}
\usepackage{caption}
\usepackage{amssymb,amsmath,amsthm}
\usepackage{mathrsfs}
\usepackage{mathtools}
\usepackage{hyperref} 
\usepackage{cleveref} 

\theoremstyle{plain}
\newtheorem{theorem}{Theorem}[section]
\newtheorem{lemma}[theorem]{Lemma}
\newtheorem{prop}[theorem]{Proposition}
\newtheorem{cor}[theorem]{Corollary}

\theoremstyle{remark}
\newtheorem{defn}{Definition}
\newtheorem{assump}{Assumption}

\newcommand{\cf}{\mathrm{CF}}
\newcommand{\dsp}{\mathrm{SP}}
\newcommand{\popsize}{\eta} 
\newcommand{\Ef}{{\mathrm{I}\!\mathrm{E}}}
\newcommand{\E}[1]{\Ef\!\left(#1\right)}
\newcommand{\Pf}{{\mathrm{I}\!\mathrm{P}}}
\renewcommand{\P}[1]{\Pf\!\left\{#1\right\}}

\DeclarePairedDelimiter\abs{\lvert}{\rvert}
\makeatletter
\let\oldabs\abs
\def\abs{\@ifstar{\oldabs}{\oldabs*}}
\makeatother

\definecolor{olive}{rgb}{0.42, 0.56, 0.14}
\usepackage{xcolor}

\usepackage[normalem]{ulem}
\newcommand{\com}[1]{} 
\definecolor{changecolor}{RGB}{192,64,0}

\begin{document}

\begin{frontmatter}
\title{Identifying circular orders for blobs in phylogenetic networks}

\author[1]{John A. Rhodes}\ead{jarhodes2@alaska.edu}
\author[2]{Hector Ba\~{n}os\corref{cor1}}\ead{hector.banoscervantes@csusb.edu}
\author[3]{Jingcheng Xu}\ead{xjc@stat.wisc.edu}
\author[3,4]{C\'ecile An\'e}\ead{cecile.ane@wisc.edu}

\cortext[cor1]{Corresponding author}
\affiliation[1]{
  organization={Department of Mathematics and Statistics, University of Alaska - Fairbanks},
  postcode={99775-6660}, state={AK}, country={USA}}
\affiliation[2]{
   organization={Department of Mathematics, California State University, San Bernardino},
   postcode={92407}, state={CA}, country={USA}}
\affiliation[3]{
   organization={Department of Statistics, University of Wisconsin - Madison},
   postcode={53706}, state={WI}, country={USA}}
\affiliation[4]{
  organization={Department of Botany, University of Wisconsin - Madison},
  postcode={53706}, state={WI}, country={USA}}

\author{}

\begin{abstract}
Interest in the inference of evolutionary networks relating species or
populations has grown with the increasing recognition of the importance of 
hybridization, gene flow and admixture, and the 
availability of large-scale genomic data.
However, what network features may be validly inferred from various data types
under different models remains poorly understood.
Previous work has largely focused on level-1 networks, in which reticulation
events are well separated, 
and on a general network's tree of blobs,
the tree obtained by contracting every blob to a node.
An open question
is the identifiability of the topology of a blob
of unknown level.
We consider the identifiability of the circular order
in which subnetworks attach to a blob, first proving
that this order is well-defined for outer-labeled planar blobs.
For this class of blobs,
we show that the circular order information from 4-taxon subnetworks
identifies the full circular order of the blob.
Similarly, the circular order from 3-taxon rooted subnetworks
identifies the full circular order of a rooted blob.
We then show that subnetwork circular information is identifiable
from certain data types and evolutionary models.
This provides a general positive result for high-level networks,
on the identifiability of the ordering in which taxon blocks attach to blobs
in outer-labeled planar networks.
Finally, we give examples of blobs with different internal structures which cannot
be distinguished under many models and data types.
\end{abstract}

\end{frontmatter}

\section{Introduction}

The statistical inference of evolutionary relationships between taxa
is a substantial challenge for phylogenetic methods when both
incomplete lineage sorting and
hybridization (or events such as introgression, admixture or gene flow)
have occurred.
Many basic questions remain about 
what features of a network it is even possible to infer  in a statistically
consistent manner, under the Network Multispecies Coalescent (NMSC)
as well as simpler stochastic models.
The theoretical precursor of this question is that of \emph{parameter identifiability}:
from the distribution of some data type 
under some generative model, what parameters, or functions of parameters can be determined?

While progress has been made in recent years, much of this has been in the setting of level-1 networks,
where the cycles in species networks caused by hybridizations are disjoint
\cite{Solis-Lemus2016,Banos2019,GrossEtAl2021,2023XuAne_identifiability,2024ABGarrotelopesR}.
Equivalently, the biconnected components, or \emph{blobs},
have only 1 reticulation in level-1 networks.
See also \cite{Ardiyansyah2021_level2} for some partial results in the level-2 case.

{Beyond the level-1 case, little is known about network identifiability
from biological data, although much has been accomplished in a more theoretical framework.
For instance, it is known that certain types of networks can be reconstructed
from displayed trees, 
distance information, 
trinets, or other phylogenetic structures
\cite{willson2006,Willson2011,vanIerselMoulton2014,vanIersalMoulton2018,BordewichEtAl2018,SempleToft2021}.
However, how, or even whether, this source information can be robustly
inferred from biological data has generally been left unaddressed.
From a biological perspective, different loci can evolve at different
rates and also be affected by incomplete lineage sorting (typically modeled with a
coalescent process) resulting in gene trees not displayed in the species network.
In such a setting, how one can extract a robust
network signal must be carefully considered.}

\medskip

In this work, {with inference from biological data in mind},
we push beyond level-1 to consider the class of \emph{outer-labeled planar} networks, in which the network can be
drawn in the plane with no edge crossings, with the taxa lying on the ``outside'' of the drawing.
While level-1 networks have this property, networks of higher level may also.
{For blobs in these networks, we first prove the graph-theoretical fact that there is a unique
circular order in which taxon blocks must attach around the blob when
drawn in the plane.
This order might be considered as the ``outer'' structure of a blob.

We next show that this order for a blob is identifiable from certain information
on the subnetworks induced by subsets of 4 taxa. Specifically, 
suppose that
for each choice of 4 taxa from distinct blocks around the blob we know that the 
induced subnetwork 
has either a blob 
with a known circular order of the 4 taxa,
or instead has an edge that, when removed, disconnects the 4 taxa into
two known groups of 2 taxa.
Then this information uniquely determines the blob's circular order.

Combined with recent results on inferring the tree of blobs of a network
from similar  information (\cite{2023XuAne_identifiability,ABMR2022}, and a new result in \Cref{ssec:logdet} of this work), this determines information
about  planar embeddings of the full network. We do not obtain results on the
``inner'' structure of the blobs, but since we require only that blobs be outer-labelled planar, our result is quite  general.}

To connect this {mathematical} identifiability result to data, we consider several data types
that have been studied or used for practical network inference.
These are average genetic distances \cite{2023XuAne_identifiability},
quartet concordance factors \cite{Solis-Lemus2016,Banos2019,ABR2019},
and genomic logDet distances \cite{ABR2022-logdetNet} for ultrametric networks.
For each, we  show the 4-taxon information needed to apply our result can be
obtained under several statistical models of gene tree formation
and sequence evolution.
Gene tree models range from the  simplest ``displayed tree'' model,
to a coalescent model on displayed trees, to the full coalescent model on a
network in which lineages behave independently.
Sequence evolution models are those commonly adopted in phylogenetics,
although some analyses make further assumptions such as constant
mutation rates, or restrictions on edge lengths in blobs.
{Since our analysis proceeds by analyzing individual blobs in the network, they apply to all blobs
that are outer-labeled planar, even if other blobs in the network are not.
We believe these are the first detailed general
results on identifiability from biological data
of a network's topology beyond the level-1 case.}

Finally, for several of these data types we construct new examples of
networks whose structures are not fully identifiable.
These include blobs which may be large (involving more than 4 taxon blocks)
and of high level (involving more than one hybridization).
While the existence of such examples is a negative result, the examples are an important
contribution to understanding the limits of which network topological features
may be identifiable, which is vital to advancing practical inference.

\medskip

{
Our work leaves many questions unaddressed. While we consider very general networks,
delineating large classes of networks for which stronger identifiability results
apply remains important, especially considering robust inference from biological data.
When full identifiability is lacking, it is valuable to extract identifiable
characteristics of the network (such as a blob's circular order),
as they may provide biological insight.
A deeper study of the classes of indistinguishable networks
for different data types is also needed, so that alternative
networks leading to the same data distribution can be better understood.
Finally, although we consider three biologically-derived data types in this work,
others might be used, either on their own or in conjunction with
those considered here, to estimate the network structure.}

\section{Phylogenetic networks and blobs}

\subsection{Networks}\label{ssec:networks}
We adopt the definitions concerning networks from \cite{2024Ane-anomalies},
but introduce terminology informally here.

A \emph{rooted topological phylogenetic network} $N^+$ on a set of
taxa $X$ is the basic object of interest.
Such a network is a rooted directed acyclic graph (DAG).
We require the leaves to be of degree 1.
An edge incident to a leaf is called \emph{pendent}. 
The edges of $N^+$ are partitioned into \emph{hybrid edges} which
share child nodes with at least one other \emph{partner} hybrid edge,
and \emph{tree edges} which do not share child nodes.
Its nodes are either the \emph{root}, \emph{hybrid nodes} which are children of
hybrid edges, or \emph{tree nodes}.
A network is \emph{binary} if the root has degree 2 and all other internal nodes
have degree 3.
The leaves of $N^+$ are bijectively labeled by elements of $X$.

A network $N^+$  may be given a metric structure by assigning
to each edge $e$ a pair of parameters
$(\ell(e),\gamma(e))$, where $\ell(e)\ge 0$ is an \emph{edge length}
with $\ell(e)> 0$ for tree edges, and
$\gamma(e)\in(0,1]$ is a \emph{hybridization} or \emph{inheritance parameter}
which sums to 1 over all partner edges with a common child node.
For any tree edge $e$ this means $\gamma(e)=1$.
Although it is sometimes desirable to allow pendent edges to have 0 length, e.g. to include
multiple individuals from the same population, or extinct taxa ancestral to extant taxa, for simplicity 
we rule that out. Nonetheless our results generalize to that situation with straightforward adjustments.

We say a node or edge is \emph{ancestral to} or \emph{above}  another in
$N^+$ if there is a, possibly empty, directed path from the first
to the second.
{An \emph{up-down path} is an undirected
path joining nodes $(u_1,u_2,\dots,u_n)$ 
such that for some $i$, $u_i\dots u_2u_1$ and $u_i\dots u_{n-1}u_n$
are directed paths in $N^+$.}

The \emph{least stable ancestor} (LSA) of the taxa on a network
$N^+$ is the lowest node through which all directed paths from the root
to any taxon must pass. Under standard models, most sources of data,
as well as the quartet information we assume we can access in
Section \ref{sec:genomicdata}, give no information
about the structure of $N^+$ above the LSA.
We therefore consider the LSA network induced
by $N^+$, as
obtained by deleting all edges and nodes strictly ancestral to the LSA,
and rerooting at the LSA.
We make the following assumption for the remainder of this work.

\begin{assump}\label{assump:LSA}
The network $N^+$ has no edges above its LSA,
that is,
its LSA is its root.
\end{assump}

For a network $N^+$ on $X$, a network
$N^+_1$ on $X_1\subseteq X$ is \emph{displayed} in $N^+$ if
it can be embedded in $N^+$ with a label-preserving map such that
edges of $N^+_1$ map to edge-disjoint directed paths in $N^+$
(see \cite{2024Ane-anomalies} for more details).
In particular, a tree $T^+$ on $X$ is displayed in $N^+$ if it can be
obtained by deleting all but one parent edge from each hybrid
node in $N^+$,
along with remaining edges no longer on up-down paths between taxa,
and additional edges and nodes so $T^+$ satisfies \Cref{assump:LSA}.
We generally do not suppress degree-2 nodes so each edge in $T^+$ maps to
a single edge in $N^+$.

For a tree $T$ on $\{a,b,c,d\}$, $T$ \emph{displays the quartet}
$ab|cd$ if $T$ has an edge that, when removed,
disconnects $T$ and partitions the leaves as $ab|cd$.
A network $N^+$ displays $ab|cd$ if there exists a tree $T$
displayed in $N^+$, that displays $ab|cd$.

The \emph{semidirected network} $N^-$ associated to $N^+$ is
the unrooted network obtained from the LSA network
induced by $N^+$ by undirecting all tree
edges while retaining directions of hybrid edges, and suppressing the LSA
if it is of degree 2. 
Edges and nodes in $N^-$ inherit classifications as hybrid or tree
from those of $N^+$. When we refer to a semidirected network,
we always assume it is obtained from a rooted phylogenetic network in this way. In particular,
it is always possible to root a semidirected network, directing all non-hybrid
edges away from the root consistently with the already directed hybrid edges.
{Up-down paths in $N^+$ correspond exactly to undirected paths
in $N^-$ that have no v-structure, that is, no segment $u_{i-1}u_iu_{i+1}$
in which $(u_{i-1}u_i)$ and $(u_{i+1}u_i)$ are directed hybrid edges
\citep{2023XuAne_identifiability}.
Thus we may use the terminology ``up-down path'' in the semidirected setting.}

\medskip

The results obtained in later sections apply to both rooted and
semidirected networks so we refer to a network as $N$ without superscript
to reflect this generality, unless stated otherwise with
more specific $N^+$ or $N^-$ notation.
However, in proving results that apply to semidirected networks,
it is often more efficient to work
in the context of rooted networks.
One of the key concepts used in our arguments, the funnel of Definition~\ref{def:funnel},
is inherently a rooted notion, and can differ for different rootings of a
semidirected network. Others, such as the lowest nodes of blobs of
Definition~\ref{def:lowest}, can be defined in semidirected networks
(i.e, they are invariant to a choice of root), but are simpler to describe in a rooted setting.
{Since, by definition, any semidirected network arises from a rooted one,
we may choose to only consider rooted networks when convenient.}

\subsection{Blobs}

A \emph{cut node} in a graph is a node whose deletion
disconnects the graph.
A \emph{cut edge} in a graph is an edge whose deletion disconnects the graph.
It is common to define blobs in networks as maximal connected subgraphs with
no cut edges (i.e., as maximal \emph{2-edge-connected} subgraphs).
For this work we adopt a stricter definition.
A \emph{blob} is a maximal subgraph with no cut nodes, that is, a
maximal \emph{biconnected} subgraph.
A \emph{trivial blob} is a blob with at most one edge,
so that cut edges are considered trivial blobs.

\begin{figure}
\centering
\includegraphics[scale=1]{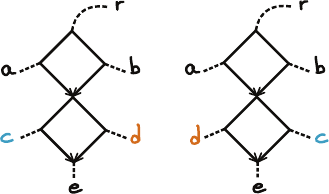}
  \caption{
  Two planar embeddings of a network. Under the definition of a blob
  adopted in this work, as a biconnected component, the central structure
  is composed of two blobs, whose circular orders of articulation nodes are the
  same in both embeddings.
  If blobs were defined as 2-edge-connected components, there would be
  only one blob but the embeddings would give different circular orders.
}\label{fig:2edge}
\end{figure}

The differences in these definitions is illustrated in \Cref{fig:2edge}.
By our definition, using biconnectedness,
the left figure shows two blobs of four edges each.
Under  the 2-edge-connected definition, all eight edges form a single blob.
Our interest is in possible orders of taxa
around planar embeddings of the graph, and the right figure shows an alternative
planar embedding of the same graph.
Viewed as two blobs, the orders of taxa around each blob are unchanged,
but viewed as a single blob there are two distinct orders.
Our results are simpler to state by adopting the biconnected definition.
In general, biconnected components are subgraphs of
2-edge-connected components, so we work with a finer decomposition of a network.
However, for binary networks the two notions of non-trivial blobs coincide,
although trivial blobs differ (cut edges versus cut nodes).

For a binary network $N$, its \emph{tree of blobs} $T$ is
obtained by contracting each non-trivial blob into a single node
\cite{Gusfield2007,2023XuAne_identifiability,ABMR2022}.
Since $N$ is binary, distinct non-trivial blobs do not share
any node and correspond to distinct nodes in $T$.
The edges of $T$ correspond to cut edges of $N$, so $T$ is roooted if $N$
is rooted, and $T$ is unrooted if $N$ is semidirected.
The \emph{reduced} tree-of-blobs is obtained by suppressing
any degree-2 node in $T$, other than its root if $T$ is rooted.

Since a phylogenetic network has no directed cycles,
we use the term \emph{cycle} to refer to a subnetwork which forms a cycle when
all edges are undirected. Any cycle either is or lies in a non-trivial blob.

An \emph{articulation node} of a blob $B$ in $N$ is a node
in $B$ that is a leaf or incident to an edge of $N$ that is not in $B$.
A blob with $m$ distinct articulation nodes is called an \emph{$m$-blob}.
Assumption \ref{assump:LSA} implies that if the LSA of  $N$
is in a non-trivial blob, then it is not an articulation node of that blob,
even though it would have been if $N$ had additional edges above the LSA.

\section{Structure of general blobs}

In this section, we define less standard terminology, and establish some facts
about arbitrary blobs which will be useful in later sections.

A subset of nodes that will play a special role
in our arguments is the following.
\begin{defn} \label{def:lowest}
A \emph{lowest node} of a blob is a node in the blob which has no descendent nodes in the blob.
\end{defn}

The following result characterizes a non-trivial blob's lowest nodes, and shows that
they are well-defined on a semidirected network because
they are independent of the root location.
\begin{lemma}\label{lem:lowest}
Let $B$ be a non-trivial blob in a rooted phylogenetic network.
Then the lowest nodes of $B$ are precisely the hybrid articulation nodes
of $B$ which have no descendent edges in $B$.
\end{lemma}
\begin{proof} 
Let $v$ be a lowest node of $B$. Since the network's leaves are the
only nodes with no descendent edges, and they are in trivial blobs,
$v$ must have one or more descendent edges.
Since $v$ is lowest, these descendent edges are not in $B$, and hence
$v$ is an articulation node of $B$.
If $v$ were not hybrid, then it would have exactly one parental edge,
which is in $B$.
But then the parent of $v$ would be a cut node of $B$,
contradicting its biconnectedness.

Conversely, if $v\in B$ has no descendent edges in $B$,
then $v$ is a lowest node of $B$.
\end{proof}

As an immediate consequence,
if a network is binary then the lowest nodes of a non-trivial blob {$B$} are
precisely its hybrid articulation nodes: For a binary network,
any hybrid articulation node of {$B$} has a single descendent edge,
which cannot be in {$B$}.

\medskip
 
We also use the notion of LSAs for blobs $B$ of
a rooted network $N^+$, as the lowest node in
$N^+$ through which all paths from the root to any node in $B$ must pass. While 
this LSA need not be in $B$, \emph{a priori},  the following shows it is.

\begin{lemma}
The LSA of a blob $B$ in a rooted phylogenetic network lies in $B$.
\end{lemma}
\begin{proof}
A node $v$ in a blob $B$ is said to be an \emph{entry node} of $B$
if there exists an edge not in $B$ with child $v$. $B$ has at most one entry node, since 
if two exist, picking directed paths
from the network root through those nodes would allow us to 
contradict that $B$ is a maximal biconnected subgraph.

If $B$ has no entry nodes then the network's root is in
$B$ and is $B$'s LSA.
If $B$ has only one entry node, it is
the LSA of the blob. 
\end{proof}

In a binary network, if a non-trivial blob's LSA and the network's root
(assumed to be the network's LSA) are the same,
that node is not incident to an edge outside the blob,
and so is not an articulation node of the blob.
If the LSAs of the blob and the network differ, however, the LSA of the blob is an
articulation node. For  a binary network, such an LSA must be a tree node.

\begin{defn} A \emph{trek cycle} (or \emph{up-down cycle}) in a DAG
is a pair $(L,R)$ of directed paths from a node $w$
(the \emph{source}) to a node $v$ (the \emph{sink}),
with no nodes other than $w$ and $v$ in common.
\end{defn}
A trek cycle  $(L,R)$ forms a cycle in the usual sense by concatenating the reversal of $L$ with $R$.
In a phylogenetic network the sink node in a trek cycle must be hybrid.

 
The following, Lemma~10 of \cite{2024Ane-anomalies},
identifies useful substructures in blobs to simplify later arguments.

\begin{lemma} \label{lem:trekcycle}
For a non-trivial blob $B$ in a DAG,
every edge or node in $B$ is contained in a trek cycle within $B$.
Thus $B$ is a union of (non-disjoint) trek cycles.
\end{lemma}

Although the proof of the following is trivial, we state it for use in later arguments.
\begin{lemma}\label{lem:abovehybrid}
Every node in a blob $B$ in a rooted phylogenetic network
is ancestral to at least one lowest node of $B$.
\end{lemma}

This lemma implies that every blob has at least one lowest node,
and hence every non-trivial blob has at least one hybrid articulation node.

\begin{defn}\label{def:funnel} 
The \emph{funnel} of a lowest node $v$ of a blob $B$ is the set of
edges and nodes in $B$ that are ancestral to $v$ but
not to any other lowest node of $B$.
\end{defn}

{Note that the term ``funnel'' has been applied in phylogenetic network studies in several ways
differing from our usage \cite{2023XuAne_identifiability,Ardiyansyah2021_level2} .} 
The notion used here is only for rooted networks,
as a funnel may depend on the root position.
The funnel of a lowest node may contain both tree and hybrid edges and nodes,
and tree or hybrid articulation nodes, as shown in \Cref{fig:linkaugment}.
Furthermore, it may contain an edge without containing the parental node of that edge,
so the funnel is not typically a subgraph.
However, deleting a funnel from a blob does give a subgraph.

\begin{lemma}\label{lem:funcon} Suppose a blob $B$ {in a rooted phylogenetic network} has at least two lowest nodes.
Then  the deletion from $B$ of the funnel of a lowest node gives a non-empty connected graph.
\end{lemma}

\begin{proof}
Let $v$ be the lowest node whose funnel was removed, and $B'$ the resulting graph.
The LSA $\rho$ of $B$ is in $B'$ since it lies above all
lowest nodes of $B$. Let $u\in B$ be any other node in $B'$.
Then there is a directed path $P$ from $\rho$ to $u$ in $B$.
But since $u$ is ancestral to a lowest node other than $v$ and all edges in $P$
are above $u$, they are also ancestral to another lowest node.
Thus all edges of $P$ are in $B'$.
\end{proof}

If a blob has a single lowest node, then by Lemma \ref{lem:abovehybrid}
its funnel is the entire blob.
Otherwise we can classify the blob's lowest nodes into two types,
illustrated in \Cref{fig:linkaugment}.

\begin{defn} Suppose a blob $B$ has at least two lowest nodes,
  one of which is $v$. Then  $v$ \emph{links} $B$ if the deletion
  of the funnel of $v$ from $B$ results in a graph with two or more
  (possibly trivial) blobs.
  If a lowest node does not link $B$ (that is, the deletion of its
  funnel results in a single blob) we say it \emph{augments} $B$.
\end{defn}

\begin{figure}
\centering
\includegraphics[scale=1.3]{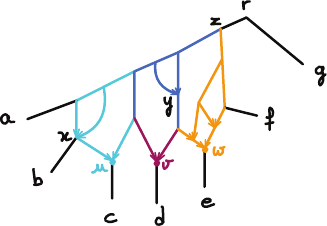}
  \caption{
  Network with one non-trivial blob $B$ (all colored edges)
  with three lowest nodes, $u$, $v$ and $w$ (colored dots).
  Articulation nodes that are not lowest can be tree nodes (e.g.~node adjacent to $a$) or hybrid nodes (e.g.~$x$, non-binary node).
  The funnel of each lowest node is shown with matching color.
  Lowest node $v$ links $B$, since deletion of its funnel from $B$ gives two blobs
  (with lowest nodes $u$ and $w$) joined by a cut edge.
  Node $w$ similarly links the blob, since deleting its funnel leaves one blob
  with lowest nodes $u,v$ and one trivial blob, the blue edge incident to  $z$.
  Lowest node $u$ augments the blob, since deleting its funnel leaves a single blob.
  This example is planar for clarity, though planarity is not assumed in the
  definition of linking and augmenting lowest nodes.
}\label{fig:linkaugment}
\end{figure}

\begin{lemma} \label{lem:augment} {In a rooted phylogenetic network, every}
  blob with 2 or more lowest nodes
  has at least one lowest node that augments it.
\end{lemma}
\begin{proof}
Suppose every lowest node of a blob $B$ links it. 
Deletion of the funnel for any lowest node $v$ leaves a connected graph by
Lemma~\ref{lem:funcon}. Since the lowest nodes are linking, these graphs have
one or more cut nodes.
Associate to each such lowest node $v$ of $B$ one of the lowest such
cut nodes $w_v$ resulting from the funnel deletion, and let $w_v$ also denote
the node in $B$ that induces it.
Of the set of associated cut nodes for all lowest nodes, choose a lowest one,
$\tilde w = w_{\tilde v}$,
with $\tilde v$ a lowest node of $B$ associated to $\tilde w$. 

When the funnel of $\tilde v$ is deleted from $B$,  the node $\tilde w$ 
becomes a cut node above one or more  blobs. But since $\tilde w$ 
is the lowest such cut node, it can be above only one blob, $B'$, to which it is incident.
Since $\tilde w$ is not in the funnel of $\tilde v$,
it must lie above some lowest node $u\ne \tilde v$ in $B$.
Then $u$ must be hybrid, $B'$ must be non-trivial and $u$ must be in $B'$.

By Lemma \ref{lem:trekcycle}, $\tilde w$ 
is in a trek cycle in $B$. Since $\tilde w$ 
becomes a cut node when the funnel of $\tilde v$ is removed,
this trek cycle includes some edge in the funnel. But then
the sink of the trek cycle must be in the funnel, and hence above $\tilde v$.
Thus $\tilde w$ is above both $\tilde v$ and $u$.

Deleting the funnel of $u$ from $B$, then, cannot delete $\tilde w$ 
or any  node in the funnel of $\tilde v$, and hence only deletes nodes in $B'$.
Thus deleting the funnel of $u$ from $B$ can only produce cut nodes from nodes below $\tilde w$. 
But such a cut  node would be lower than $\tilde w$ ,
which is a contradiction to the choice of $\tilde w$
as a lowest node in the set $\{w_v { \mid } \,v \mbox{ lowest in } {B}\}$.

Thus $B$ has an augmenting lowest node.
\end{proof}

In {\Cref{sec:orderFromQuartets},} we will use deletion of funnels of lowest nodes
in an inductive proof. The following Lemma describes the impact 
on the number of lowest nodes.

\begin{lemma}\label{lem:counthcn} {In a rooted phylogenetic network, let}
  $B$ be a blob with $n\ge 2$ lowest nodes. Then the deletion of a
 funnel of any lowest node produces a graph with exactly $n-1$ lowest nodes among its blobs.
\end{lemma}
\begin{proof}
 Let $B'$ be the subgraph of $B$ obtained by deleting the funnel of a 
 lowest node $v$.  The $n-1$ lowest nodes of $B$ other than $v$ are not in $v$'s 
 funnel, hence are in blobs of $B'$
and remain lowest nodes. 

Suppose there is a lowest node $w$ of $B'$ that is not a lowest node of $B$. 
Then all descendent edges of $w$ in $B$ are in the funnel of $v$, and hence $w$ lies 
above $v$ and no other lowest node of $B$. Therefore $w$ is in the funnel of $v$, and not in $B'$. 
\end{proof}

\section{Circular orders and outer-labeled planar blobs}\label{sec:circOrdandOLP}

We are interested in natural orders of taxa around 
a phylogenetic network embedded in the Euclidean plane, so that
as a plane graph, vertices are distinct and edges only meet at their end-points.
Generally, networks 
do not have a unique order in which
their taxa are arranged, since at any cut node one can ``rotate'' one part of
the network out of the plane, reversing the order of its taxa,
as was illustrated in \Cref{fig:2edge}.
However, the individual blobs of such a network are associated to a unique order,
as we establish in this section.
We seek an order of blocks of taxa around a blob, but first define these blocks.

\begin{defn}\label{def:taxonblock}
  For a blob $B$ in a network 
  on $X$, the \emph{taxon blocks} associated with $B$ are the
  non-empty subsets of $X$ on connected components
  obtained by deleting all edges in $B$.
\end{defn}
The taxon blocks for a blob are in bijective correspondence to the articulation
nodes of the blob, according to the articulation node through which any undirected
path from the blob to a taxon in the block must pass.
This correspondence relies upon \Cref{assump:LSA}, so that the blob's LSA
is not considered an articulation node when it coincides with the network LSA.

\begin{defn}
  A \emph{circular order} for a finite set $Y$ is an order $(y_1,y_2,\dots,y_k)$
  of its elements up to reversal and cyclic permutation.
  That is, as a circular order, $(y_1,y_2,\dots,y_{k-1},y_k)$ is the same as
  $(y_k,y_{k-1},\dots,y_2,y_1)$ and as $(y_2,y_3,\dots,y_k,y_1)$.
\end{defn}

Viewing the articulation nodes as ``labeled'' by their taxon blocks,
a circular order of the blocks arises from any mapping of the blob into the plane
which places the labels on the ``outside''.
So that such an order is not completely arbitrary,
we require that the mapping be an embedding.
To make this precise, recall that a graph embedded in the plane divides it into
connected regions called \emph{faces}.
Exactly one of these is an unbounded region, called the \emph{unbounded face}.
To simplify later wording, we adopt the following terminology.

\begin{defn}
If a planar graph $G$ is embedded in the plane as $H$, the \emph{frontier} of $H$
is the set of edges and nodes forming the topological boundary of the unbounded face.
\end{defn}

The frontier of an embedded planar graph depends upon the embedding, and not just
the abstract graph, as illustrated in \Cref{fig:2olp}.

\medskip

The following is a slight generalization of a definition of \cite{HusonRuppScorn},
allowing general subsets of vertices rather than the leaf set.

\begin{defn} \label{def:outer}
Let $G$ be a planar graph, embedded in the plane as $H$, and
$V_0$ a subset of the vertices of $G$. Then the embedded $H$ is
\emph{outer-labeled planar for $V_0$} if all elements of $V_0$ lie in the frontier of $H$.
If $G$ has any outer-labeled planar embedding for $V_0$, we say $G$ is
\emph{outer-labeled planar} for $V_0$.
A phylogenetic network $N$ is \emph{outer-labeled planar} if
$N$ is outer-labeled planar for its leaf set, labeled by $X$.
A  blob $B$ in a phylogenetic network is \emph{outer-labeled planar}
if $B$ is outer-labeled planar for its set $V_0$ of articulation nodes,
labeled by taxon blocks.
\end{defn}

\begin{figure}
 \centering
 \includegraphics[scale=1.2]{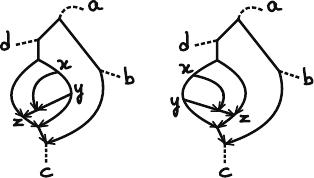}
 \caption{A planar blob $G$ may have several different  outer-labeled planar embeddings.
 Here two such embeddings are shown for the same blob, with $V_0$ consisting of articulation nodes
 adjacent to $\{a,b,c,d\}$.
 On the left, hybrid node $z$ lies in the frontier but tree nodes $x$ and $y$ do not,
 while these are reversed on the right.
 However, the circular order of articulation nodes is the same for both embeddings.
}\label{fig:2olp}
\end{figure}

If $N$ is outer-labeled planar, then each of its blobs must be,
by restricting the embedding to the blob.
Conversely, if all its blobs are outer-labeled planar, then the network must be,
as one can construct an appropriate embedding by ``gluing'' those for each blob
at articulation nodes.

An outer-labeled planar blob may have several different
outer-labeled planar embeddings, with different sets of edges and nodes
in their frontiers, as shown in \Cref{fig:2olp}.
Nonetheless, the following establishes a fundamental commonality to all such embeddings.

\begin{theorem} \label{thm:uniqcirc}
Let $G$ be an undirected  biconnected graph which is outer-labeled planar for
a subset $V_0$ of its vertices.
Then the frontiers of all the planar embeddings of $G$ with $V_0$ in the frontier
induce the same circular order of $V_0$.
\end{theorem}
\begin{proof}
Let $F$ and $F'$ denote the subgraphs of $G$ forming the frontiers of two
outer-labeled planar embeddings, $H$ and $H'$, of $G$ for $V_0$.
Since $G$ is biconnected, by Proposition 4.2.6 of \cite{Diestel2017},
$F$ and $F'$ are cycles.
Thus $F$ and $F'$ induce circular orders $\sigma$ and $\sigma'$ of $V_0$.
If $|V_0|< 4$, $\sigma=\sigma'$ is trivial, so we henceforth assume $|V_0|\ge 4$.

Let $\sigma=(x_1,x_2,...,x_n)$
and suppose, for the sake of contradiction, $\sigma\ne \sigma'$.
Without loss of generality, we may assume that $x_1,x_2$ are adjacent in $\sigma$ but not in $\sigma'$.
Then there are four elements of $V_0$ with circular order $(x_1,x_2,x_i,x_j)$ induced from $\sigma$
and $(x_1,x_i,x_2,x_j)$ from $\sigma'$.

There are two vertex-disjoint paths, $P$ and $Q$ in $F$, such that $P$
passes from $x_1$ to $x_2$ and no other node in $V_0$, and $Q$
passes from $x_i$ to $x_j$ and possibly other nodes in $V_0$
but neither $x_1$ nor $x_2$.
Let $P'_1$ and $P'_2$ be the two distinct paths in $F'$ going
from $x_i$ to $x_j$, passing through $x_1$ and $x_2$
respectively, which are disjoint except for their endpoints,
$x_i$ and $x_j$.
Note that $P$ connects $P_1'$ to $P_2'$ and does not pass
through $x_i$ or $x_j$.  Let $\alpha$ be the last node on $P$
that is also on $P'_1$ and $\beta$ the first node on $P$ after
$\alpha$ that is on $P'_2$ (see \Cref{fig:proofuniq}, path
in blue).  
Let $P_3$ be the subpath of $P$ from $\alpha$ to $\beta$.

Consider the following 3 paths from $\alpha$ to $\beta$ that intersect only
at their endpoints:
the subpaths $Q'_1,Q'_2$ of the cycle $F'$ passing through $x_i, x_j$ respectively,
and the subpath $P_3$ of $P$ in $F$.
Under the $H'$ embedding, the path $P_3$ lies in the region of the plane
bounded by $F'=Q'_1\cup Q'_2$
(see \Cref{fig:proofuniq}).
Thus by Lemma 4.1.2 of \cite{Diestel2017},
every path from $x_i$ to $x_j$ in $G$ must intersect $P_3$,
in the planar embedding $H'$.
In particular, $Q$ must intersect $P_3$ in $H'$.
Since $H'$ is planar, $Q$ and $P_3$ must intersect at a vertex,
which is a contradiction since by definition they are vertex-disjoint.
This contradiction yields $\sigma= \sigma'$.
\end{proof}

\begin{figure}
\centering
\includegraphics[scale=1.3]{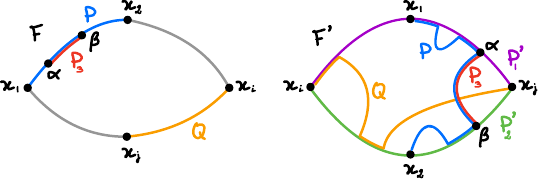}
\caption{
Schematic drawing of the paths used in the proof of Theorem \ref{thm:uniqcirc} to show each
outer-labeled biconnected graph $G$ induces a unique circular order on the labeled nodes.
Left: frontier $F$ of an embedding $H$ of $G$. Two vertex-disjoint paths of $F$,
$P$ from $x_1$ to $x_2$ (in blue), and $Q$ from $x_i$ to $x_j$ (in orange)
are highlighted. Additionally, a sub-path $P_3$ of $P$ is shown in red.
Right: frontier $F'$ (in purple and green) of an embedding $H'$ of $G$.
Part of the frontier $F$ (paths $P$ and $Q$ in blue and orange, respectively)
from the embedding $H$ is also highlighted.
Between the vertices $\alpha,\beta$ are 3 paths that are disjoint except for
their endpoints: two in $F'$ and one in $F$ ($P_3$ in red).
But then $Q$ must intersect $P_3$ in the plane.
}\label{fig:proofuniq}
\end{figure}

\begin{cor}\label{cor:uniquecirc}
  There is a unique circular order of the articulation nodes of an outer-labeled
  planar blob of a phylogenetic network,
  induced from every outer-labeled planar embedding of the blob.
\end{cor}

\medskip

If a network $N$ on an $n$-taxon set $X$ has an $n$-blob, then each taxon block
for this blob must be reduced to a single taxon, mapping each articulation node
to a single taxon in $X$. Then by \Cref{cor:uniqcirctilde}, all
outer-labeled planar embeddings of $N$ induce the same circular order of $X$,
which leads us to the following definition.
\begin{defn}\label{def:circnetwork}
  Let $N$ be an outer-labeled planar network on an $n$-taxon set $X$.
  If $N$ has an $n$-blob, the \emph{circular order of $N$} is defined as the
  unique circular order of $X$ induced by all outer-labeled planar
  embeddings of $N$.
  More generally,
  $N$ is said to be \emph{congruent} with a circular order $\sigma$ of $X$ if there exists an
  outer-labeled planar network $N'$ on $X$ displaying $N$ which has an $n$-blob and
  circular order $\sigma$.
\end{defn}

{It is easy to see that any outer-labeled planar network $N$ with $n\ge 4$ taxa but lacking
an $n$-blob is congruent with several circular orders.
Indeed, pick a cut node or edge separating the taxa into two sets of size $\ge 2$, and
`rotate' that part of the embedding on one side
(as in \Cref{fig:2edge}) to get a new embedding. These embedded networks can be completed to be outer-labelled planar with $n$-blobs, by adding edges between pendant edges, and then have distinct circular orders.}

Note this definition differs from that {of a circular order} defined
in terms of split systems, which is used in the development of splits graphs
\cite{Semple2005,HusonRuppScorn,2017Gambette-uprooted}.
Indeed, the network types differ, as ours are explicit and splits graphs are
implicit.
While we do not use split systems here, the following shows a connection for 4-taxon trees.


\begin{lemma}\label{lem:circ4taxontree}
  Let $T=ab|cd$ be a resolved 4-taxon tree. Then $T$ is congruent with exactly two
  circular orders: $\sigma_c=(a,b,c,d)$ and $\sigma_d=(a,b,d,c)$.
\end{lemma}
\begin{proof}
 Consider the network with a single 4-cycle, where $a$ descends from the
 hybrid node and, moving around the cycle  $b,c,d$ are, in order, attached by edges.
 This network has circular order $\sigma_c$ and displays $T$.
 Interchanging the taxa $c,d$ gives a network with order $\sigma_d$ which also displays $T$.

  To show $T$ is not congruent with $\sigma_b=(a,c,b,d)$ we reason similarly
  to the proof of Theorem \ref{thm:uniqcirc}.
  Suppose $T$ were congruent to $\sigma_b$, and $N$ is a 4-blob network displaying
  $T$ with that order. Fix an outer-labeled planar embedding of $N$,
  with $F$ the frontier of the blob.
  In the embedded $T$, let $p_{ab}$ be the unique path from $a$ to $b$ (considered
  undirected), and $p_{cd}$ the unique path from $c$ to $d$.
  Since the order along $F$ is $\sigma_b$, paths from $a$ and $b$ {to the 4-blob} meet $F$ on different segments
  between the articulation nodes for $c$ and $d$: let $F_a$ be the segment for $a$ and $F_b$
  that for $b$.
  Then we can find nodes $\alpha\in F_a$ and $\beta\in F_b$ along $p_{ab}$
  and three disjoint paths between
  $\alpha$ and $\beta$: two segments, $F_1$ and $F_2$, of $F$
  and a subpath, $p_{\alpha\beta}$, of $p_{ab}$ that intersects $F$ only at $\alpha,\beta$.
  Since $T=ab|cd$, $p_{ab}$ and $p_{cd}$ do not intersect, so $\alpha$ and
  $\beta$ must be distinct from the articulation nodes for $c$ and $d$.
  Since $p_{\alpha\beta}$ lies in the region of
  the plane bounded by $F_1$ and $F_2$, every path in $N$ between $c$ and $d$
  must intersect $p_{\alpha\beta}$ by \cite[Lemma 4.1.2]{Diestel2017},
  a contradiction since $p_{ab}$ and $p_{cd}$ do not intersect.
  Therefore $T$ is not congruent with $\sigma_b$.
\end{proof}

\section{Circular orders from quartets}
\label{sec:orderFromQuartets}

Having established that an outer-labeled planar blob in a phylogenetic network
has a unique circular order of its articulation nodes, we turn to the question
of whether this circular order can be determined from information on
the network's induced quartet networks.
A positive answer will provide a pathway for showing that the circular order
can be consistently inferred from certain approaches to empirical data analysis.

\medskip

{By an induced rooted network on a subset of taxa we mean the network obtained from a rooted network $N^+$ by
retaining only those edges and nodes ancestral to at least one of the taxa, removing 
edges and nodes above the LSA, and then suppressing degree-2 nodes. Alternatively, it is obtained by 
retaining only edges and nodes on up-down paths between taxa, and suppressing degree-2 nodes. An
 induced semidirected network on a subset of taxa is obtained either from the induced rooted network, or from $N^-$ by 
 retaining edges and nodes on up-down paths between taxa, and suppressing degree-2 nodes.}

The following notion is similar to that of a 4-taxon set that defines
an edge in a tree, which is used in the definition of a dense set of quartets
for a tree \cite{Gambette2012}.

\begin{defn}
For a blob $B$ in a phylogenetic network,
a 4-taxon set is \emph{$B$-informative} if its elements are in
distinct taxon blocks associated to $B$.
A collection of 4-taxon sets is a \emph{full $B$-informative} collection
if, for each choice of 4 blocks associated to $B$,
one of the 4-taxon sets contains a taxon from each block.
\end{defn}

Given a $B$-informative 4-taxon set, it is easy to see that
modulo 2-blobs the topology of the induced quartet network is dependent
only on the blocks of the taxa and not the individual taxa themselves.
{Indeed, the paths from the four taxa to the blob reach it at articulation nodes determined by 
the taxon blocks, and the induced quartet network topologies can only differ by 2-blobs along these paths}.

\smallskip

The simplest example of the connection between blob structure and circular order
for a 4-taxon network, is captured by the following.

\begin{lemma}\label{lem:4taxonset-displayedsplits}
Let $N$ be a binary network on $\{a,b,c,d\}$. {Then one of the following holds:}
\begin{enumerate}
\item
  If $N$ has a cut edge separating $a,b$ from $c,d$,
  it displays only the quartet $ab|cd$.
  If $N$ is also outer-labeled planar, it is congruent with
  circular orders $(a,b,c,d)$ and $(a,b,d,c)$.
\item
  If $N$ has a 4-blob
  then it displays at least 2 quartets.
  If the 4-blob of $N$ is also outer-labeled planar with
  unique circular order $(a,b,c,d)$, then $N$
  displays exactly 2 quartets,
  $ab|cd$ and  $ad|bc$.
\end{enumerate}
\end{lemma}
\begin{proof}
  In the first case, the first statement is trivial
  since any cut edge in $N$ must be in all displayed trees.
  To show that $N$ is congruent with $\sigma_c=(a,b,c,d)$,
  we build a network ${N}'$ displaying $N$,
  outer-labeled planar, with a 4-blob and order $\sigma_c$. 
  Since an edge incident to a leaf can be directed towards the leaf,
  we can build ${N}'$ from $N$ by adding a hybrid edge
  from the edge incident to {$d$}  to the edge incident to $a$.
  The other order is handled similarly.
  
  Now consider the second case, with a 4-blob. By \cite[Lemma~1]{2024Ane-anomalies}, $N$
  must display at least 2 distinct quartets.
  If $N$ is outer-labeled planar,
  by Corollary \ref{cor:uniquecirc} it has a unique circular order $(a,b,c,d)$.
  Then by \Cref{lem:circ4taxontree},
  the displayed quartets can only be $ab|cd$ and $ad|bc$.
\end{proof}

\Cref{lem:4taxonset-displayedsplits} is extended to non-binary networks
in the Appendix, \Cref{lem:4taxonset-displayedsplits-non-binary},
which includes a third case for
networks having a central node whose deletion disconnects all 4 taxa.
Such networks display the star tree only, and, like the 4-taxon star tree,
are congruent with all 3 circular orders.
These results imply the following remarkable connection.

\begin{cor}\label{cor:circofdisplayedtrees}
  Let $N$ be a 4-taxon outer-labeled planar network.
  $N$ is congruent with a circular order $\sigma$ if and only if
  all its displayed trees are congruent with $\sigma$.
\end{cor}

\smallskip

We assume that we can access the following information about induced quartet networks.
Some possible sources of such information are discussed in the next section.

\begin{defn}\label{def:orderinfo}
For a binary outer-labeled planar blob $B$ in a phylogenetic network
and a full $B$-informative collection of 4 taxon sets,
the \emph{4-taxon circular order information} is the collection of 4-taxon orders
specified by the induced outer-labeled planar quartet
networks for the 4-taxon sets.
More specifically, this information consists of:
\begin{enumerate}
\item\label{item:cut} for a quartet network without a 4-blob,
and hence a cut edge inducing a split $xy|zw$ of the taxa, 
the circular orders are $(x,y,z,w)$ and $(x,y,w,z)$.
\item\label{item:4blob} 
for a quartet network with a 4-blob,
the circular order is the unique order induced from any outer-labeled planar embedding
of the quartet network.
\end{enumerate}
\end{defn}

\begin{theorem}\label{thm:circorder}
Let $B$ be a binary outer-labeled planar blob in a phylogenetic network.
Then 4-taxon circular order information on a full $B$-informative set
of 4-taxon sets determines
the unique circular order of its articulation nodes induced by all
outer-labeled planar embeddings of the blob.
\end{theorem}

In the proof below, we in fact use only information from 4-taxon sets that
induce a 4-blob (item \ref{item:4blob} in Definition \ref{def:orderinfo}).
While order information from quartet networks without 4-blobs is potentially
useful for improved data analyses,
it is redundant for establishing identifiability of the blob's circular order.

Also, while our argument establishes identifiability of the
circular order, it does not suggest an efficient algorithmic way of obtaining the order.
We only show that the exact 4-taxon order information is compatible with no 
circular order than  the one arising from an outer-labeled planar embedding of the full blob.

\begin{proof}[Proof of Theorem \ref{thm:circorder}]
{As noted in \cref{ssec:networks}, we may assume the network is rooted.} 
We proceed by induction on the number $n$ of lowest nodes of the blob.

In the base case of $n=1$,  there is only one lowest node $v$. If there are fewer than 4 articulation nodes, the result is trivial. Otherwise,
for any three articulation nodes $x_1,x_2,x_3$ other than $v$, the induced quartet network on taxa chosen from the taxon blocks for $v,x_1,x_2,x_3$ has a 4-blob.
This is because by Lemma \ref{lem:abovehybrid} all edges in the blob are ancestral to $v$,
so no edges are lost in passing to the induced network.
Thus by assumption we know the circular suborder of $\{v,x_1,x_2,x_3\}$ for every such set. From these it is straightforward to deduce the unique circular order on the full set of articulation nodes.

Now let $n\geq 2$ and assume the result for blobs with $n-1$ lowest nodes. Again the result is trivial unless there are at least 4 articulation nodes, so we consider only that case.
By Lemma \ref{lem:augment} there is at least one lowest node $v$ that augments the blob.
Let $B'$ be the subgraph of $B$ obtained by deleting the funnel of $v$,
so $B'$ is a blob on the induced network $N'$ on all
non-descendents of $v$'s funnel edges.
By Lemma \ref{lem:counthcn}.
$B'$  has exactly  $n-1$ lowest nodes, all inherited from $B$. 
Thus a unique circular order of the articulation nodes of $B'$ is determined by induction.
 
We now must show the articulation nodes of $B$ on the funnel of $v$ can be uniquely placed to extend the order for $B'$.
(Note however that we do not know which node is $v$, nor which articulation nodes lie on $B'$ or the funnel of $v$;
we simply have a full $B$-informative
collection of sets of 4 taxa which determines the unique circular
order on some subset of the articulation nodes of $B$ which includes all those on the unknown $B'$, and we want to show that we have enough information to extend
the circular order to the full set of articulation nodes.)

We consider two cases:

\smallskip

If $B'$ has only one lowest node, $u$, then taxa from the blocks for  $u,v$ and any two other articulation nodes $x,y$ of
$B$ induce a quartet network with a 4-blob. This is because 
every edge in $B$ must be ancestral to $u$ or $v$ by Lemma \ref{lem:abovehybrid}, so passing to an induced quartet
network on taxa in the blocks for these four articulation nodes retains all edges in the blob, and thus produces no cut nodes from the blob.
Since we have circular orders for each such choice of taxa, varying $x,y$ is
enough to determine a full circular order for the blob.
Note that in this case we did not need to use the already-determined
circular order on $B'$.

 \smallskip
 
If $B'$ has at least 2 lowest nodes, we consider two subcases
depending on the location of $v$ in the frontier of the embedded blob.
In the first, $v$ is between neighboring lowest nodes
$u,w$, with the LSA of the blob $L$ not between $u$ and $v$ or between $v$ and $w$.
In the second{,} $L$ is an articulation node,
there are no lowest nodes between $v$ and $L$ on
one side, and on the other the neighboring lowest nodes are $u$ and then $w$
(see \Cref{fig:updownfrontier}).

\begin{figure}
\centering
\includegraphics[scale=1.2]{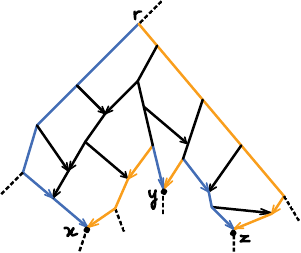}
\caption{
The frontier of an outer-labeled planar embedding of a blob is composed of an
alternating succession
of down- and up-paths (shown in alternating blue and orange),
with one pair for each lowest node.
With 3 lowest nodes in this blob $x$, $y$ and $z$ (dots),
there are 3 down- and 3 up-paths.
Along these paths are all four combinations of tree/hybrid articulation/non-articulation nodes.
Node $y$ illustrates subcase 1 in the proof of Theorem~\ref{thm:circorder},
as a lowest node that augments the blob and is
is separated from the blob's LSA by lowest nodes along the frontier.
Node $z$ illustrates subcase 2 as a augmenting lowest node between the LSA
of the blob and another lowest node along the frontier.
} \label{fig:updownfrontier}
\end{figure}

In the first  subcase, with $v$ between $u$ and $w$,
the lowest nodes  $u,v,w$ and any fourth articulation node $x$ on $B$  lead to an induced quartet network with a 4-blob. To see this, note
this quartet network includes the three maximal biconnected sets $B_u,B_v,B_w$ containing $u,v,w$ but composed only of edges ancestral to $u,v,w$. These also contain the funnels of $u,v,w$.
Moreover, since $B_u,B_v$ (respectively $B_v,B_w$) share the edge(s) parental to the node at which the up-paths in the frontier of $B$ from 
$u,v$  (respectively $v,w$) meet, $B_u\cup B_v\cup B_w$ is biconnected. Thus $u,v,w$ are articulation nodes on a single blob.
 
Varying the articulation node $x$ among the others on $B'$ (if any exist) shows
that the location of $v$ in the order from $B'$ can be determined.
Then varying $x$ among those on the funnel of $v$ shows we can determine  which are on the $u$- and $w$-sides of $v$.
If there is more than one articulation node, say $x_1,x_2$ on $B$ between $u$ and $v$, then $u,v,x_1,x_2$ induce a quartet network with a 4-blob, since these are all articulation nodes on the blob composed of edges ancestral to $u$ or $v$ (again using that in $B$ there is an edge ancestral to both $u$ and $v$). The order of $u,v,x_1,x_2$ for all such $x_1,x_2$ between $u$ and $v$ is sufficient to order all articulation nodes between $u$ and $v$. Similarly, the articulation nodes between $v$ and $w$ can be ordered from quartets, completing this case.

In the second subcase, where $v$ is between the blob LSA and neighboring
lowest node $u$, with $w$ the other lowest node neighboring $u$ we proceed similarly.
Considering $v,u,w$ and any fourth articulation node $x$ on $B$ leads
to an induced quartet network with a 4-blob, as above. Varying $x$ among
articulation nodes on $B'$ determines the location of $v$ with
respect to all those of $B'$ from quartet orders.
Then varying $x$ among those articulation nodes on the funnel of $v$ determines
whether they lie on the LSA side or the $u$ side of $v$.
Then considering induced blobs on $v,u,x_1,x_2$ where the $x_i$ are on a common
side of the funnel of $v$, is sufficient to obtain the full order.
\end{proof}


\section{Circular order from genomic data}
\label{sec:genomicdata}

The previous section's main result, Theorem \ref{thm:circorder},
states that the order of branching of subnetworks around a blob can be identified
from quartet information for an outer-labeled planar network.
To relate this to identifiability from biological data, we
discuss here several frameworks in which the 4-taxon circular order information of
Definition \ref{def:orderinfo} is itself identifiable,
from various data types and model assumptions.

\begin{defn}[gene tree models]
\label{def:gene-tree-models}
 We consider the following models for how gene trees arise from a given
 rooted phylogenetic network $N^+$.
 To generate gene trees with branch lengths in substitutions per site,
 each edge $e$ in $N^+$ is assigned
 a length $g(e)\geq 0$ in number of generations
 (with $g(e)>0$ for tree edges),
 an effective population size $\popsize(e)>0$
 and a mutation rate $\mu(e)>0$ per site per generation.
 \begin{enumerate}
 \item \emph{Displayed tree} (DT) model:
  gene trees are drawn from the set of
  trees displayed in $N^+$, with displayed tree $T$ having probability
  equal to the product of inheritance probabilities of all edges in $T$:
  $$\P{T} = \gamma(T) = \prod_{e\in T}\gamma(e).$$
  Each edge $e$ in $T$ maps to a full edge in $N^+$ and is assigned
  length $\ell(e)=\mu(e) g(e)$. 
 \item \emph{Network multispecies coalescent model with
independent inheritance} (NMSCind): gene trees evolve according to the coalescent
  model within each population. At a hybrid node with parental edges
  $e_1,\dots,e_m$, 
  each lineage is inherited from population $e_k$ ($1\leq k\leq m$) with
  probability $\gamma(e_k)$, 
  independently of the other lineages \cite{Luay2012}.
Each gene tree edge $e'$ formed in network edge $e$ is assigned length $\ell(e')=\mu(e)\,g(e')$. 
 \item \emph{Network multispecies coalescent model with
common inheritance} (NMSCcom): gene trees evolve according to the coalescent
  model within each population. At a hybrid node with parental edges
  $e_1,\dots,e_m$, 
  all lineages of a given gene are inherited from the same population $e_k$
  ($1\leq k\leq m$), chosen with probability $\gamma(e_k)$, and thus form within trees displayed on the network.
  \cite{2011GerardGibbsKubatko}.
The length of a gene tree edge is assigned as in  item 2. 
  \end{enumerate}
\end{defn}

Under the DT model, there is no incomplete lineage sorting: lineages
coalesce immediately at the time of speciation.
The traditional NMSC model is what we denote NMSCind.
A network coalescent model with correlated inheritance, depending on a parameter $\rho\in[0,1]$, was defined in
\cite{2023fogg_phylocoalsimulations}, with both NMSCind ($\rho=0$) and NMSCcom ($\rho=1$) as submodels.

These gene tree models are combined with molecular substitution models,
to allow for genetic sequence data. Rate variation across lineages can arise from
variable sustitution rates $\mu$, and also from variable generation times,
captured by a variable number of generations across different paths
between the same pair of nodes.
To model rate variation across genes, we further consider that each gene
has a relative rate $R>0$ drawn independently from some distribution, with mean 1.

More nuanced models could be considered, in which the population size $\popsize(e)$
and mutation rate $\mu(e)$ vary continuously over time instead of being
considered as constant within an edge (and variable across edges).
Such models would require heavier notations, but can be approximated in
our framework with the addition of nodes in the network to subdivide each
edge into multiple edge segments with different parameters.

Each of the following subsections considers a different type of data from these models, applying the earlier results of this work to show circular orders for blobs are identifiable.
Since our focus is on identifiability,
and not on inference from data, we consider only expectations associated to a data type.
However, for each one, the required information
can be consistently estimated from a sample of gene trees or
from multilocus molecular sequences,
assuming both gene sequence length and number of genes approach infinity.
Thus each is potentially useful in practical inference.

\subsection{Average distances from genes}\label{subsec:avdist}

The \emph{average genetic distance}, or expected path length between two taxa
on gene trees, is a metric that was studied in \cite{2023XuAne_identifiability},
and shown to be useful for identifying the tree of blobs of a network.
This distance, 
in units of substitutions per site, is formally defined
for a metric species network $N^+$ and one of the three gene tree
models $M$ above by
\begin{equation*}
  \mathfrak{D}_{N^+}^M(x,y) = \E{ R \, D_T(x,y) },
\end{equation*}
where the expectation is taken over random gene trees $T$ (unscaled by gene-specific rates) generated
on $N^+$ under model $M$, and over
gene-specific relative rates $R$.
Since $R$ is independent of $T$, with mean 1,
this simplifies to
\begin{equation}\label{eq:avdist}
  \mathfrak{D}_{N^+}^M(x,y) = \E{ D_T(x,y) }.
\end{equation}

\smallskip

For any set of 4 taxa $a,b,c,d$ and for a distance $D$ on these taxa,
the \emph{4-point distance sums-of-pairs} is the vector
\[\dsp = \dsp_{abcd} (D)= (\dsp_{ab|cd}, \dsp_{ac|bd}, \dsp_{ad|bc})\]
where 
\[\quad \dsp_{ab|cd} = \dsp_{ab|cd}(D) =D(a,b)+D(c,d) .\]

Recall that a distance $D$ is a \emph{tree distance} if there exists a tree $T$ with edge lengths $\ell(e)$,
such that for all taxon pairs $x,y$,
$$D(x,y)=D_T(x,y) = \sum_{ e\in\mathrm{path}_T\; x\leftrightarrow y} \ell(e).$$
Distances on $T$ are exactly those satisfying the 4-point condition \cite{Semple2005}:
for any quartet $xy|zw$ displayed in $T$,
$$\dsp_{xy|zw} < \dsp_{xz|yw} = \dsp_{xw|yz}.$$

\smallskip
We introduce a notion analogous to that of a (hard)-anomalous network as defined
in \cite{2024Ane-anomalies} in the context of quartet concordance factors.

\begin{defn}
A network $N^+$ is \emph{average-distance-anomalous} under a model
if  there are two quartet trees, say $ab|cd$ and $ac|bd$,
with the first displayed on $N^+$ and the second not,
yet for the average genetic distance
$$\dsp_{ab|cd} > \dsp_{ac|bd}.$$
\end{defn}

From the 4-point condition for trees one might naively suppose
that the SP corresponding to some displayed quartet tree would produce the
smallest SP on the quartet network. An
average-distance-anomalous network is one contradicting that supposition.

Under the DT and NMSCcom models, no network is average-distance-anomalous,
as follows from the next proposition.
However, for the NMSCind model a network may be anomalous.
Although we delay giving an example of such a network until the end of this
section, this necessitates the use of the following
assumption.

\begin{assump}\label{assump:noAnomAD} (NoAnomAD)
Under the model NMSCind of gene tree evolution, the metric network
is such that none of its induced 4-taxon networks are average-distance-anomalous.
\end{assump}

\begin{prop}\label{prop:4taxadist}
Under each of the models DT, NMSCcom, and NMSCind+NoAnomAD
on a rooted metric binary network $N^+$,
average genetic distances determine the order information of \Cref{def:orderinfo}
for each of its outer-labeled planar blobs $B$.
Specifically, each induced quartet network for  a $B$-informative set displays exactly 1 or 2  tree topologies, and, with taxa $a,b,c,d$:
 \begin{enumerate}
 \item\label{item:claim1AD} If only $ab|cd$ is displayed, then
$SP_{abcd}=(p,q,q)$ with $p<q$ and the quartet network on $a,b,c,d$
has a cut edge separating $a,b$ from $c,d$; and circular orders
$(a,b,c,d)$ and $(a,b,d,c)$.
 \item \label{item:claim2AD} 
If $ab|cd$ and $ac|bd$ are displayed, then $SP_{abcd}=(p,q,r)$, with
$p,q< r$ and the quartet network has a 4-blob with circular order $(a,b,d,c)$.
\end{enumerate}
\end{prop}

\begin{proof} Let $a,b,c,d$ be the taxa on $N^+$. 
By Lemma \ref{lem:4taxonset-displayedsplits}, we may assume either
$N^+$ has a 4-blob, with circular order $(a,b,c,d)$, or
has a cut edge separating $a,b$ from $c,d$, and hence has orders
$(a,b,c,d)$ and $(a,b,d,c)$.

  For the DT and NMSCcom model, the average distance
  on $N^+$ can be expressed as
  a weighted sum over displayed trees, denoted $S^+$ here (for ``species'' tree):
 \begin{equation}\label{eq:Dsum}
 \mathfrak{D}_{N^+}^M(x,y) = \sum_{S^+}\gamma(S^+)\; \mathfrak{D}_{S^+}^M(x,y).
 \end{equation}
  Under the DT model, $\mathfrak{D}_{S^+}^\mathrm{DT} = D_{S^+}$,
  and under the NMSCcom model, the distance  on each displayed tree
  is from the MSC on $S^+$ as a species tree:
  $\mathfrak{D}^{\mathrm{NMSCcom}}_{S^+}=\mathfrak{D}_{S^+}^\mathrm{MSC}$.
 
Recall that all tree edges of $N^+$ have positive length.
Moreover, when degree-2 nodes are suppressed in a displayed tree $S^+$,
each resulting edge arises from conjoining several edges,
the highest of which must be a tree edge,
yielding only positive length edges.
In analyzing the above sum, then, we assume all
edge lengths of $S^+$ are positive.

By \Cref{eq:Dsum}, it is enough to show claim \ref{item:claim1AD} in the case $N^+=S^+$ is a tree.
For the DT model, this is immediate from the 4-point condition, while for the NMSCind model it results from the NoAnomAD assumption.
To show it holds for the MSC observe that species tree $S^+$ is, without loss of generality, either
asymmetric, $(((a, b), c), d)$,
or symmetric, $((a, b), (c, d))$.
In both cases let $e$ be the edge above only $a, b$ and in the second case
$\tilde e$ the edge above $c, d$.
Under from the MSC on $S^+$, consider the event $E$ that a coalescence
occurs on $e$ or $\tilde e$ (if $S^+$ is symmetric).
Conditional on $E$, a random gene tree $T$ must have topology $ab|cd$, so that
$$\mathbb E(D_T (a, b) + D_T (c, d) | E) < \mathbb E(D_T (a, c) + D_T (b, d) | E) =\mathbb  E(D_T (a, d) + D_T (b, c) | E).$$ 
Conditional on the complement $\bar E$, the $a,b,c$ lineages reach a common node and become
exchangeable under the coalescent model, so
$$\mathbb E(D_T (a, b) + D_T (c, d) | \bar E) = \mathbb E(D_T (a, c) + D_T (b, d) | \bar E) = \mathbb E(D_T (a, d) + D_T (b, c) | \bar E) .$$
Claim \ref{item:claim1AD} then follows.

For \ref{item:claim2AD} , by Lemma~\ref{lem:4taxonset-displayedsplits}, if $N^+$ has a 4-blob
it displays exactly 2 quartets, $ab|cd$ and $ac|bd$.
Applying claim 1 to each displayed tree, then, we see that
$\dsp$ is a convex sum of triples of the form $(s_i,S_i,S_i)$ and $(S_i,s_i,S_i)$
where $s_i<S_i$ for each. Thus the last entry of $\dsp$,
corresponding to $ad|bc$, is strictly largest, establishing claim 2.
\end{proof}

\begin{prop}\label{prop:ToBdist} 
  Let $N^+$ be an outer-labeled planar metric rooted binary
  phylogenetic network.
  Under the DT, NMSCcom, and NMSCind+NoAnomAD models,
  the reduced unrooted tree of blobs of $N^+$
  is identifiable from average genetic distances.
\end{prop}

\begin{proof} 
  For the DT model,
  we follow arguments from \cite{2023XuAne_identifiability} who defined the
  distance split tree $T'$ constructed from $\mathfrak{D}_{N^+}^\mathrm{DT}$ and
  the associated sums of distance pairs
  $\dsp = \dsp(\mathfrak{D}_{N^+}^\mathrm{DT})$
  across all subsets of $4$ taxa.
  From their Theorem~8, $T'$ is a refinement of $T$,
  the reduced unrooted tree of blobs of $N^+$.
  Therefore, we simply need to prove that any edge $e$ in $T'$ is also in $T$. 
  To prove this, we consider the split $A|B$ induced by $e$ in $T'$.
  Suppose, by contradiction, that $e$ is not an edge in $T$.
  By Lemma 3.1.7 in \cite{Semple2005}, $e$ refines a node in $T$.
  In $N^+$, this node corresponds to an $m$-blob with $m\geq 4$,
  so we can select
  taxa $\{a_1,a_2,b_1,b_2\}$ with $a_1,a_2\in A$ and $b_1,b_2\in B$
  such that $N^+_{a_1,a_2,b_1,b_2}$ has a $4$-blob.
  By the definition of the distance split tree \cite{2023XuAne_identifiability},
  $A|B$ must satisfy the 4-point condition for
  $\mathfrak{D}_{N^+}^\mathrm{DT}$, so we must have that
  $\dsp_{a_1a_2b_1b_2} = (s,S,S)$ with $s\leq S$.
  But since $N^+_{a_1,a_2,b_1,b_2}$ is outer-labeled planar,
  $\dsp_{a_1a_2b_1b_2}$ has a single largest entry by
  Proposition~\ref{prop:4taxadist}:
  a contradiction. Hence, $T=T'$ can be identified from average
  distances.

  Now suppose $M$ is the NMSCcom
  or NMSCind+NoAnomAD model.
  Proposition \ref{prop:4taxadist} implies that
  $\Sigma(\mathfrak{D}_{N^+}^M)$,
  the set of quartet splits that satisfy the 4-point
  condition for metric $\mathfrak{D}_{N^+}^M$,
  is the same under $M$ as it is under the DT model.
  The distance split tree $T'$ depends on the input distance $D$ only
  through $\Sigma(D)$. Therefore under $M$, $T'$ will be the same and
  the rest of the argument for the DT model applies.
\end{proof}

This result extends Corollary~11 in \cite{2023XuAne_identifiability}
to outer-labeled planar networks of level higher than 1,
and to the NMSCcom and NMSCind+NoAnomAD models,
although with the restriction that the network is
binary.

Combining Propositions~\ref{prop:4taxadist}, \ref{prop:ToBdist} and
Theorem \ref{thm:circorder} we obtain the following.

\begin{cor}
  Let $N^+$ be an outer-labeled planar metric rooted binary phylogenetic
  network.
  Then the reduced unrooted tree of blobs and the circular order for each blob
  on $N^+$ are identifiable from its average pairwise distances,
  under the DT, NMSCcom and NMSCind+NoAnomAD models.
\end{cor}

\begin{proof}
 The reduced unrooted tree of blobs of $N^+$ is identifiable
 from average distances by Proposition~\ref{prop:ToBdist},
 so we may focus on each blob individually by passing to the induced network
 on a subset of taxa.
 For every set of 4 taxa, we can determine the quartet order information of
 Definition \ref{def:orderinfo} from average distances
 by Proposition~\ref{prop:4taxadist}.
 Finally, applying Theorem \ref{thm:circorder} completes the proof.
\end{proof}

The claims of Proposition~\ref{prop:4taxadist} and the
results built on it do not hold for the  NMSCind model without the NoAnomAD assumption,
as the following proposition shows.

\begin{prop}
  Under the NMSCind model on an outer-labeled planar binary 
  phylogenetic network
  on 4 taxa, the two smallest entries of $\dsp$ need not indicate the circular order for a 4-blob.
\end{prop}

\begin{proof}
Consider the network $N$ in \Cref{fig:ind-inherit-counterexp} (left), 
considered as rooted along edge $e_6$, although
our choice of parameters will make $\dsp = \dsp(\mathfrak{D}_{N}^\mathrm{NMSCind})$
independent of the root location.
Assume for simplicity population sizes $\popsize(e)=1$ and mutation rates $\mu(e)=1$ for all edges $e$,
so edge lengths $g(e)=\ell(e)$ represent both coalescent units and substitutions per site.
Also assume $g(e_a) = g(e_b) + g(e_1)$.
We will show that for  sufficiently small $m$,
$$\dsp_{ab|cd} > \dsp_{ac|bd} > \dsp_{ad|bc},$$ 
which is incompatible with the
circular order $(a,b,c,d)$ in the sense that it violates
the NoAnomAD Assumption~\ref{assump:noAnomAD}.

\begin{figure}
\centering
\includegraphics[scale=1.1]{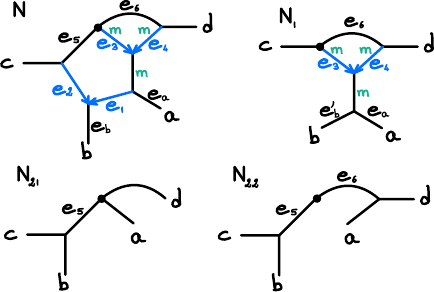}
\caption{Example network $N$ (left) on which the claims of Proposition~\ref{prop:4taxadist}
do not hold under the NMSCind model,
when the edge length $m$ is sufficiently small.
$N_1$, $N_{21}$ and $N_{22}$ are displayed in $N$.
The root (not shown) is on edge $e_6$.
}
\label{fig:ind-inherit-counterexp}
\end{figure}

  For a random gene tree $T$ from the NMSCind on $N$,
  let $E_1$ denote the event that the $b$ lineage chooses $e_1$;
  $E_{21}$ the event that $b$ chooses $e_2$ and $a$ chooses $e_3$;
  and $E_{22}$ the event that $b$ chooses $e_2$ and $a$ chooses $e_4$.
  We first show that for sufficiently small $m$, we have
  \begin{equation}
    \label{eq:E_1}
    \E{\dsp_{ab|cd}(D_T) \mid E_1} > \E{\dsp_{ac|bd}(D_T) \mid E_1} = \E{\dsp_{ad|bc}(D_T) \mid E_1}.
  \end{equation}

  Note that the distribution of $\dsp(D_T)$
  from $N$ conditional on $E_1$ is the same as
  the distribution of $\dsp(D_T)$ from the network $N_1$ displayed in $N$,
  in Figure \ref{fig:ind-inherit-counterexp}.
  The equality in \eqref{eq:E_1} follows from
  $N_1$ being symmetric with respect to $a, b$ and $g(e_a) = g(e'_b)$.

  To show the first inequality in~\eqref{eq:E_1},
  we consider the event $A$ that no lineages coalesce in the three edges
  of length $m$, and consider the
  expectation of $\dsp(D_T)$ from $N_1$ conditional on $A$.
  When $a, b$ are inherited from the same hybrid edge,
  $\{a, b, c\}$ or $\{a, b, d\}$ become exchangeable.
  By symmetry, the conditional expectation of
  $\dsp(D_T)$ then has equal entries.
  When $a$ and $b$ are inherited from different hybrid edges,
  there are several cases. When no coalescence
  happens on edge $e_6$, then $\{a, b, c,d\}$ are exchangeable
  above the root and by symmetry, the conditional expectation of
  $\dsp(D_T)$ has equal entries.
  When there is a coalescence along $e_6$, then
  $\dsp_{ab|cd}(D_T) = \dsp_{ac|bd}(D_T)$
  if $a$ is inherited from $e_4$ and $b$ from $e_3$,
  whereas
  $\dsp_{ab|cd}(D_T) > \dsp_{ac|bd}(D_T)$
  if $a$ is inherited from $e_3$ and $b$ from $e_4$.
  %
  Since all cases above have nonzero probability, we have
  \[
  \E{\dsp_{ab|cd}(D_T) \mid E_1 \cap A} > \E{\dsp_{ac|bd}(D_T) \mid E_1 \cap A}.
  \]
  On the event $E_1 \cap \bar{A}$, $T$ has topology $ab|cd$.
  In this case $|\dsp_{ab|cd}(D_T) - \dsp_{ac|bd}(D_T)|$ is twice the length
  of the internal edge of $T$, and
  $\E{|\dsp_{ab|cd}(D_T) - \dsp_{ac|bd}(D_T)| \mid E_1 \cap \bar{A}} < 4m + 2t_0$,
  where $t_0=1$ is the expected time for two lineages to coalesce.
  Hence as $m \to 0$,
  \[
    \P{\bar{A} \mid E_1}
    \E{\abs{\dsp_{ab|cd}(D_T) - \dsp_{ac|bd}(D_T)} \mid E_1 \cap \bar{A}} \to 0.
  \]
  Combining the cases for $A$ and $\bar{A}$, we have
  $\E{\dsp_{ab|cd}(D_T) \mid E_1} > \E{\dsp_{ac|bd}(D_T) \mid E_1}$ when $m$
  is sufficiently small, establishing \eqref{eq:E_1}.

 Next we show that for $j = 1, 2$, we have
  \begin{equation}
    \label{eq:E_2}
    \E{\dsp_{ab|cd}(D_T) \mid E_{2j}} = \E{\dsp_{ac|bd}(D_T) \mid E_{2j}} >
    \E{\dsp_{ad|bc}(D_T) \mid E_{2j}}.
  \end{equation}

  Conditional on $E_{21}$, the network becomes $N_{21}$ of Figure \ref{fig:ind-inherit-counterexp}, and $T$ has topology $ad|bc$ if $b$ and $c$ coalesce on
  $e_5$, in which case
  $\dsp_{ab|cd}(D_T) = \dsp_{ac|bd}(D_T) > \dsp_{ad|bc}(D_T)$.
  Otherwise, $\{a,b,c\}$ are exchangeable above $e_5$ and the
  conditional expectation of $\dsp$ has equal entries.
  This establishes \eqref{eq:E_2} for $j = 1$.
  For $j=2$, conditional on $E_{22}$, the same argument applies using network $N_{22}$ of Figure \ref{fig:ind-inherit-counterexp} with
  gene tree topology $ad|bc$ if some coalescence occurs along $e_5$ or $e_6$,
  or with exchangeability of $\{a,b,c,d\}$ above the root otherwise.

  Combining \eqref{eq:E_1} and \eqref{eq:E_2}, we have
  $\dsp_{ab|cd} > \dsp_{ac|bd} > \dsp_{ad|bc}$ as desired.
\end{proof}

\subsection{Quartet concordance factors}\label{subsec:qCF}

Quartet \emph{concordance factors (CFs)}
underlie a number of identifiability results and inference methods for level-1 networks \cite{Solis-Lemus2016, Banos2019, ABR2019, 2024ABGarrotelopesR}, and 
for the tree of blobs of general networks \cite{ABMR2022}.

Under any gene tree model, quartet CFs are the probabilities, for each 
set of 4 taxa, of the possible gene quartet trees relating them. 
Since for the three models $M$ we consider (assuming a binary network in the case of DT) only
resolved quartet trees have positive probability, 3-entry vectors suffice for these:
$$\cf_{abcd}=(\cf_{ab|cd},\cf_{ac|bd},\cf_{ad|bc})=(\mathbb P_M({ab|cd}),\mathbb P_M({ac|bd}),\mathbb P_M({ad|bc})).$$
These CFs can be obtained by marginalizing the distribution of rooted metric gene trees over edge lengths, root locations, and all but 4 taxa. The model of
sequence evolution on the gene trees is thus irrelevant as long as gene tree topologies are identifiable. 
Although CFs
can be consistently estimated from a sample of gene trees, in practice such a sample is never available,
so they are estimated from inferred gene trees
or directly from multilocus sequence data. 

\smallskip

Following \cite{2024Ane-anomalies}, we say a 4-taxon network $N^+$ is
\emph{(hard-)anomalous} when there are two quartet trees, say $ab|cd$ and $ac|bd$,
with the first displayed on $N^+$ and the second not, yet
$$\cf_{ab|cd}<\cf_{ac|bd}.$$
More informally, a network is anomalous if the displayed tree(s) are not the most probable gene trees.

Under the DT model on a 4-taxon network, only the trees displayed on the network have positive probability, so no anomaly can occur. Under the NMSCcom model,
all binary gene trees have positive probability but,
by Proposition 12 of \cite{2024Ane-anomalies}, there are again no anomalies.

Anomalies can occur under the NMSCind model, although only one embedded
network substructure (a $3_2$-cycle) is currently known to lead to them.
However, even with such a substructure, an anomaly does not occur for all
parameter values on the network.
See Proposition 1 of \cite{ABMR2022} for an extreme example of an anomalous
outer-labeled planar network displaying a single tree topology,
and \cite{2024Ane-anomalies} for deeper investigations and simulations on the
frequency of anomalies.

To identify circular orders of an outer-labeled planar network
under NMSCind from CFs, we use the following:

\begin{assump}\label{assump:noAnomQ} (NoAnomQ)
Under the model NMSCind of gene tree evolution, the metric network
is such that none of its induced 4-taxon networks are anomalous.
\end{assump}

We first show that under our three models, assuming no anomalous 4-networks,
CFs determine the quartet order information of \Cref{def:orderinfo}.

\begin{prop}\label{thm:CFdisptree}
Under each of the models DT, NMSCcom, and NMSCind+NoAnomQ
on a rooted metric binary network,
quartet CFs determine the order information of \Cref{def:orderinfo} for each of
its outer-labeled planar blobs $B$.
Specifically, each induced quartet network for  a $B$-informative set displays exactly 1 or 2  tree topologies, and, with taxa $a,b,c,d$:
\begin{enumerate}
\item\label{it:1disp} If only $ab|cd$ is displayed, then
$\cf_{abcd}=(p,q,q)$ with $p>q$ and the quartet network on $a,b,c,d$
has a cut edge separating $a,b$ from $c,d$; and circular orders
$(a,b,c,d)$ and $(a,b,d,c)$.
\item \label{it:2disp}
If $ab|cd$ and $ac|bd$ are displayed, then $\cf_{abcd}=(p,q,r)$, with
$p,q>r$ and the quartet network has a 4-blob with circular order $(a,b,d,c)$. 
\end{enumerate}
\end{prop}
\begin{proof}
From Lemma \ref{lem:4taxonset-displayedsplits}, an induced quartet network displays
either 1 or 2 tree topologies, depending on whether it lacks or has a 4-blob.
If it lacks a 4-blob,
it has a cut edge separating the taxa into
groups of two, which must be $a,b$ and $c,d$  if $ab|cd$ is displayed.

If a quartet network displays only the tree topology $ab|cd$, then under the
DT model, $\cf_{abcd}=(1,0,0)$,  since  for each displayed tree the CF is $(1,0,0)$.
Likewise,
under NMSCcom, $\cf_{abcd}=(p,q,q)$ with $p>q$,  since 
for each displayed tree the CF has this form \cite{Allman2011}. 
For NMSCind, Theorem 1 of \cite{ABMR2022} implies $\cf_{abcd}=(p,q,q)$, and  $p>q$ follows from the assumption NoAnomQ.

If a quartet network displays 2 tree topologies, $ab|cd$ and $ac|bd$,
then $\cf_{abcd}=(p,q,r)$ with $p,q>r$ is shown for the DT and
NMSCcom models by  similarly considering the displayed trees individually.
For the NMSCind model, this follows from the NoAnomQ assumption.
\end{proof}

\begin{cor} Let $N^+$ be a rooted metric
binary phylogenetic network.
Then under the models DT, NMSCcom, and NMSCind+NoAnomQ models
the reduced unrooted tree of blobs and the circular order for each
outer-labeled planar blob on
$N^+$ are identifiable from quartet CFs. 
\end{cor}
\begin{proof} Under NMSCind, the reduced unrooted  tree of blobs of
$N^+$ is identifiable from CFs by Proposition 4 of \cite{ABMR2022},
for generic parameters.
One can check that with the NoAnomQ assumption, the generic condition can be dropped,
since it is used only in the determination of 4-taxon sets
whose subnetwork has a 4-blob, which can instead be determined by Theorem \ref{thm:CFdisptree}.

Similarly, for DT and NMSCcom, Theorem \ref{thm:CFdisptree} allows the arguments
of \cite{ABMR2022} to apply, to prove the identifiability of the tree of blobs from CFs.

Theorems \ref{thm:CFdisptree} and \ref{thm:circorder} then show each outer-labeled planar blob's
circular order is identifiable.
\end{proof}

\subsection{LogDet distances from genomic sequences}\label{ssec:logdet}

In \cite{Allman2019,ABR2022-logdetNet}, logDet distances computed from genomic sequences were shown
to identify species trees and level-1 network topologies
under an independent coalescent model, provided certain technical assumptions hold.
This information differs from the average distances considered
in~\Cref{subsec:avdist}, since the logDet distance is computed from
concatenated multigene sequences, as opposed to sequences for each gene individually.
Here we study using LogDet distances to identify circular orders of blobs in outer labeled planar networks.

\smallskip

Recall
that a rooted metric network is said to be \emph{ultrametric} if all paths
from the root to a leaf have the same length.
An ultrametric network is therefore \emph{time-consistent}, in the
sense that for every node $u$, all the paths from the root
to $u$ have the same length.

In this subsection we consider only 
rooted networks that are ultrametric when edge lengths are measured in generations.
We also assume a mutation rate that is constant across all edges, so ultrametricity in substitution units holds as well.
(For simplicity, we do not investigate generalization to time-dependent mutation rates,
or to varying substitution processes across genes,
as in \cite{Allman2019,ABR2022-logdetNet}.)

The requirement that a network  be ultrametric is tied to the use of rooted triples,
rather than quartets, for identifiability.
In particular, triples of logDet distances from each displayed rooted triple network play a 
similar role to the CF triples for displayed quartets.
We thus require results and definitions paralleling those
of earlier sections, with rooted triple networks replacing quartet networks.
Since many of the necessary arguments are straightforward adaptations of those
for quartets, we omit many details and only highlight key differences.

\smallskip
 
We strengthen the definition of an outer-labeled planar blob in a rooted network by considering the
\emph{extended articulation nodes} $\widetilde V_0$ as the usual
articulation nodes together with the blob's LSA.
A blob is \emph{extended outer-labeled planar} if it is outer-labeled planar
with respect to $\widetilde V_0$.  In this subsection we use the term
\emph{$m$-blob} to refer to a blob with $m$ extended articulation nodes.

If a blob 
does not contain the network root then the blob's LSA must be an articulation node,
so `extended outer-labeled planar' is synonymous with `outer-labeled planar'.
These terms have different meanings only for blobs containing the network root,
since the extended definition requires a planar embedding with the root in
the frontier, while the standard definition does not.

Applying \Cref{thm:uniqcirc} to $\widetilde V_0$ immediately yields the following.

\begin{cor}\label{cor:uniqcirctilde}
  There is a unique circular order of the extended articulation nodes of an extended outer-labeled
  planar blob of a rooted phylogenetic network,
  induced from every extended outer-labeled planar embedding of the blob.
\end{cor}

\begin{defn}
For a blob $B$ in a rooted  phylogenetic network,
a 3-taxon set is \emph{$B$-informative} 
if its elements are in
distinct taxon blocks associated to $B$,
excluding the block associated to the blob's LSA.
A collection of 3-taxon sets is a \emph{full $B$-informative} collection
if, for each choice of 3 blocks of $B$ not associated to its LSA,
one of the 3-taxon sets contains a taxon from each block.
\end{defn}

If $B$ is an extended outer-labeled planar blob, then the
3-taxon network induced by a $B$-informative set of 3 taxa
need not have the same LSA as $B$.
Nonetheless, since 
there can only be a chain of 2-blobs between
the LSA of the 3 taxa and the LSA of $B$,
whichever LSA we refer to has no impact on the notion of circular order
information in the following.

\begin{defn}\label{def:tripordinfo}
For an extended outer-labeled planar blob $B$ with LSA $L=L_{B}$
in a binary rooted phylogenetic network
and a full $B$-informative collection of 3 taxon sets,
the \emph{3-taxon circular order information} is the collection of circular orders of $x,y,z,L$
specified by the induced
rooted triple networks for the 3-taxon sets $\{x,y,z\}$.
More specifically, this information consists of
all circular orders of $x,y,z,L$ that are compatible with all
unrooted displayed trees in the rooted 3-taxon  network:
\begin{enumerate}
\item\label{item:tcut} for a rooted triple network without a 4-blob,
and hence a cut edge inducing a split $xy|zL$ of the taxa and LSA, 
the circular orders are $(x,y,z,L)$ and $(x,y,L,z)$.
\item\label{item:t4blob} 
for a rooted triple network with a 4-blob,
the circular order is the unique order induced from any extended outer-labeled planar embedding
of the rooted triple network.
\end{enumerate}
\end{defn}

The following is an analog  for rooted triples of \Cref{thm:circorder}. Unfortunately, it cannot be immediately deduced from that quartet result, since
even extending the rooted network with an additional 
outgroup taxon, there are fewer 
collections of 3 ingroup taxa than 
there are collections of 4 taxa.

\begin{theorem}\label{thm:rtcircorder}
Let $B$ be an extended outer-labeled planar blob in a
rooted binary phylogenetic network.
Then 3-taxon circular order information on a full $B$-informative set
of 3-taxon sets determines
the unique circular order of its extended articulation nodes induced by all
outer-labeled planar embeddings of the blob.
\end{theorem}

\begin{proof} Our argument parallels that for \Cref{thm:circorder}, so we only provide a sketch.

Proceeding by induction on the number of lowest nodes of $B$,
if $v$ is the only lowest node, then for every choice of articulation nodes $x,y$ we have a unique circular order for
$v,x,y,L$ which must be congruent with the unique circular order of the full set
$\widetilde V_0$ of the extended articulation nodes.
It is straightforward to see the orders for all choices of $x,y$ determine that for $\widetilde V_0$.

Assuming the result for any blob with $n-1$ lowest nodes,
suppose that $B$ has $n$, and pick an augmenting lowest node $v$.
With $B'$ obtained from $B$ by deleting $v$ and its funnel,
the circular order for the extended articulation nodes of $B'$,
which has $n-1$ lowest nodes, is determined by induction.
Note that $L$ must be in $B'$ since $v$ is augmenting, hence $L$ is included in this order.  

If $B'$ has a single lowest node $u$, then since all edges of $B$
are ancestral to $u$ or $v$, the rooted network induced by $u,v,x$ has a 4-blob
for any other articulation node $x$ of $B$.
Letting $x$ vary over articulation nodes of $B'$ determines a unique
placement of $v$ within the circular order for $B'$.
Then, letting $x$ vary over articulation nodes of $B$ on the funnel of $v$
determines the position of each one
relative to $v$ and to articulation nodes of $B'$ (including $L$).
What remains to be determined is the position of these articulation nodes on the funnel
relative to one another.  
But  by restricting to the blob containing $v$ in the induced
network on those taxa which are descended from the funnel of $v$,
which has only one lowest node, a unique circular order for
all articulation nodes on the funnel of $v$ can also be determined.

If $B'$ has more than one lowest node, using rooted triple order information
we first determine where $v$ should be placed between other lowest nodes and
the blob LSA $L$, and then argue as in the last case to obtain the full circular order.
\end{proof}

\medskip

Recall the definition of the logDet distance between two (finite) aligned sequences of $k$ bases from taxa $a,b$ \cite{Steel94}. Let
$F_{ab}$ be the $k\times k$ matrix of relative site-pattern frequencies, whose $ij$ entry gives the proportion of sites 
in the sequences exhibiting base $i$ for $a$ and base $j$ for $b$. Let $f_a$ denote the vector of row sums of $F_{ab}$, 
and $f_b$ the vector of column sums, so that these marginalizations give the proportions of various bases in the sequences of $a$ and $b$. With $g_a$ and $g_b$ the products of the entries of $f_a,f_b$, respectively, the logDet distance is
$$d_{ab} =-\frac 1k \left ( \ln |\det(F_{ab})|-\frac 12 \ln (g_ag_b)\right ).$$
Instead letting $F_{ab}$ be the matrix of expected genomic pattern frequencies
from a model $M$ on a rooted network $N^+$, this formula yields a logDet distance determined by the model's distribution of genomic sequence data. It is from these distances that we seek to identify circular orders for blobs.
With $d^M_{xy}$ the logDet distance computed from expected genomic pattern
frequencies for taxa $x,y$ under model $M$, let
$$D_{abc}=D^M_{abc}=(d^M_{ab},d^M_{ac},d^M_{bc})$$
be the triple of pairwise distances between three taxa $a,b,c$.

\smallskip

The following is an extension of results for a level-1 ultrametric network
in \cite{ABR2022-logdetNet}, which can be proved similarly
to Theorem 1 of \cite{ABMR2022}.

\begin{prop}\label{prop:LD4}
Consider a 3-taxon ultrametric rooted binary network with LSA $L$.
For the DT, NMSCcom, and  NMSCind models with constant mutation rate
and generic numerical parameters,
if the network has a 4-blob then the three entries of $D_{abc}$ are distinct,
while if there is an internal cut edge separating $a,b$ from $c,L$
then $d_{ac}=d_{bc}$.
\end{prop}
\begin{proof}
If the network on $a,b,c$ has a 4-blob, a modification of the topological
argument in the proof of Theorem 1 of \cite{ABMR2022} shows it displays a
level-1 network with a 4-blob. Choosing hybridization parameters in such a way
that lineages are constrained to stay in this displayed network,
Theorem 1 of \cite{ABR2022-logdetNet} shows that for NMSCind the three entries of
$D_{abc}$ are distinct for generic parameters on the displayed network.
Since these entries are analytic functions of the parameters for the full network,
for generic parameters on the full network they must also be distinct.
For the DT and NMSCcom model, a similar argument applies, though showing the
analog of Theorem 1 of \cite{ABR2022-logdetNet} holds for them requires a
detailed application of the algebraic Lemma 4 of \cite{Allman2019}
to pattern frequency matrices.

If there is an internal cut edge separating $a,b$ from $c$ on the network,
then under any of these models for each possible gene tree there is an
equiprobable gene tree obtained by interchanging the $a,b$ labels,
by ultrametricity.
This implies the expected pattern frequency arrays for $a,c$ and for $b,c$ are the same,
so $d_{ac}=d_{bc}.$
\end{proof}

It is also straightforward to follow the arguments of \cite{ABMR2022},
replacing its use of the combinatorial quartet distance capturing 
topological information on quartets separating pairs of taxa
by the rooted triple distance of \cite{Rhodes2020},
to obtain the following.

\begin{prop}\label{prop:ToBrt}
  The reduced rooted tree of blobs of a rooted binary network can be determined from
  the reduced rooted trees of blobs for each of its induced rooted 3-taxon networks.
\end{prop}

The  notion of an unrooted 4-taxon anomalous network for quartet CFs has a parallel for logDet distances.
\begin{defn}
An ultrametric rooted triple network $N^+$ is \emph{logDet-anomalous} for a model $M$
if there are two rooted triple trees, say $ab|c$ and $ac|b$,
with the first displayed on $N^+$ and the second not, yet
$$d_{ab}^M>d_{ac}^M.$$
\end{defn}

This notion captures the naive idea that a taxon pair appearing as a cherry
on a displayed tree should appear to have more closely related sequences
than if they are not a cherry on any displayed tree.
This idea is used by \cite{2021JiaoYang} to define a
``species-definition anomaly'' zone, in which $a$ and $b$ are two individuals
from the same species and $c$ is an individual from a different species.

From \cite{Allman2019}, an ultrametric tree $((a,b),c)$ is never
logDet-anomalous for the MSC model since $d_{ab}<d_{ac}=d_{bc}$.
However, as shown in \cite{ABR2022-logdetNet}, there are level-1 rooted triple networks
with tree of blobs $((a,b),c)$ yet $d_{ab}> d_{ac}=d_{bc}$,
although known examples require extreme parameters.
However, the question of what network structures can lead to logDet anomalies
and how common they might be has not been investigated in depth.

To identify circular orders of extended outer-labeled planar blobs in a  binary network
from logDet distances under the NMSCind model, we will use the following:

\begin{assump}\label{assump:noAnomLD} (NoAnomLD)
Under the NMSCind model, the metric network
is such that none of its induced 3-taxon rooted networks are LogDet-anomalous.
\end{assump}

\begin{prop}\label{thm:LDdisptree} 
Under each of the models DT, NMSCcom, and NMSCind+NoAnomLD
with constant mutation rate
on a binary rooted ultrametric network with generic numerical parameters,
logDet distances determine the order information of Definition \ref{def:tripordinfo}
for each of its extended outer-labeled planar blobs $B$.
Specifically, each induced rooted triple network for a $B$-informative
set displays exactly 1 or 2 tree topologies. For taxa $a,b,c$ and blob LSA $L$:
\begin{enumerate}
\item\label{it:1disprt} If only $ab|c$ is displayed, then
$D_{abc}=(p,q,q)$ with $p<q$ and the rooted triple network on $a,b,c$
has an internal cut edge separating $a,b$ from $c$; and circular orders
$(a,b,c,L)$ and $(b,a,c,L)$.
\item \label{it:2disprt}
If $ab|c$ and $ac|b$ are displayed, then $D_{abc}=(p,q,r)$, with
$p,q<r$ and the quartet network has a 4-blob with circular order $(a,b,L,c)$.
\end{enumerate}
\end{prop}
\begin{proof} Attaching a new outgroup taxon at the root
and applying Lemma \ref{lem:4taxonset-displayedsplits} shows
an induced rooted triple network displays
either 1 or 2 tree topologies, depending on whether it lacks or has a 4-blob.
Moreover, the network must have an internal cut edge $e$ separating $a,b$ from $c$
if only $ab|c$ is displayed.

If a rooted triple network displays only the tree topology $ab|c$,
then for all three models \Cref{prop:LD4} shows $D_{abc}=(p,q,q)$.
For the DT and NMSCcom model one can show $p<q$ using
the algebraic Lemma 4 of \cite{Allman2019} with an analysis of pattern frequency arrays
on each displayed tree, similar to Theorem  8 of that work.
For NMSCind, that $p<q$ follows from the assumption NoAnomLD.

If a rooted triple network displays 2 tree topologies, $ab|c$ and $ac|b$,
then $D_{abc}=(p,q,r)$ by  \Cref{prop:LD4}. That $p,q<r$ is shown for the DT and
NMSCcom models by considering the displayed trees individually and using Lemma 4 of \cite{Allman2019}.
For the NMSCind model, that $p,q<r$ follows from the NoAnomLD assumption.
\end{proof}

\begin{cor}
  Let $N^+$ be a rooted metric binary ultrametric phylogenetic network,
  with generic parameters.
  Then under the models DT, NMSCcom, and NMSCind+NoAnomLD models,
  the reduced rooted tree of blobs and the circular order for each
  extended outer-labeled planar blob on $N^+$ are identifiable
  from logDet distances.
\end{cor}
\begin{proof}
The tree of blobs for each triple of taxa is identifiable by \Cref{prop:LD4}, and hence for the full network by \Cref{prop:ToBrt}.  \Cref{thm:LDdisptree} and \Cref{thm:rtcircorder}
give the identifiability of circular orders for extended outer-labeled planar blobs.
\end{proof}


\section{Identifiability limits beyond the circular order}

{We next exhibit} outer-labeled planar networks with distinct topologies
that have the same pairwise average genetic distances, quartet CFs, or logDet distances.
We similarly show that distinguishing between outer-labeled planar
networks and non outer-labeled planar networks is not always possible via these data types.

Some such non-identifiability results have been found previously.
Even level-1 4-taxon network topologies are not always distinguishable from one-another
\cite{Solis-Lemus2016,Banos2019,2023XuAne_identifiability}.
In any network, a 2-blob may be replaced with a single
tree edge without affecting quartet CFs nor average distances
\cite{2024Ane-anomalies,2023XuAne_identifiability}.
Also 3-blobs may be shrunk into a single 3-taxon subtree
without affecting average distances \cite{2023XuAne_identifiability}, and in some level-1 cases
be shrunk or have their hybrid node moved without
affecting CFs \cite{2024ABGarrotelopesR}. 

Our new examples of indistinguishable networks have larger blobs,
including ones of high level.
While some of these are outer-labeled planar
and thus have the same identifiable circular order,
others are not.
The fact that they are nonetheless associated with a unique circular order
suggests that it may be possible to meaningfully generalize the
notion of circular order beyond the class of outer-labeled planar networks.

{Most importantly, these examples underscore the importance
of determining how computational feasibility and
statistically valid inference from biological data
can be balanced when the network structure
is not assumed to be simple. What classes of networks can be identified from what data types remains an open problem.}

\subsection{Limits of identifiability from average distances}

We present here a class of 5-taxon networks that
are indistinguishable from a network with a single reticulation,
based on average distances under the DT model and NMSCcom models.
This includes the networks in \Cref{fig:CFnetworks}, but also others
such as in \Cref{fig:5-sunlet-labeled}.


\begin{prop}\label{prop:5-blob-unident-AD}
  Let $M$ be the DT or NMSCcom model.
  Let $N$ be a metric binary network on $\{a,b,c,d,h\}$ such that:
  \begin{enumerate}
  \setlength{\itemsep}{0pt}
  \setlength{\parskip}{0pt}
  \item $N$ contains a 5-blob.
  \item The subnetwork $N_{\{a,b,c,d\}}$ contains a cut edge that
    induces the split $ad|bc$.
  \item The degree-2 nodes in $N_{\{a,b,c,d\}}$ at which paths to $h$ last
    leave $N$ all lie on an up-down path from $a$ to $b$,
	  and all are incident to two cut edges in $N_{\{a,b,c,d\}}$.
\end{enumerate}
  If the edges incident to $a$ and $b$ in $N$ are sufficiently long, then the
  average distances on $N$ also fit a 5-sunlet $N_1$ as shown in
  \Cref{fig:5-sunlet-labeled} left:
  $\mathfrak{D}_{N}^M = \mathfrak{D}_{N_1}^M$.
\end{prop}

\begin{figure}[h]
  \centering
  \includegraphics[scale=1]{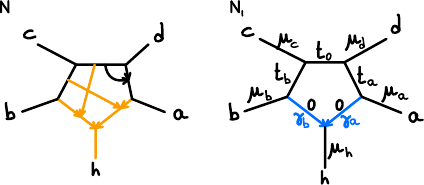}
  \caption{ 
  $N$ (left) is an example network satisfying the conditions of
  \Cref{prop:5-blob-unident-AD}. It is
  indistinguishable from the level-1 network $N_1$ (right) using average distances
  under the DT or NMSCcom model
  provided the pendent edges to $a$ and $b$ are sufficiently long.
  Orange edges of $N$ are 
  absent from $N_{abcd}$, which is shown in black.
  $N_1$ is a \emph{$5$-sunlet}, that is, has 5 taxa and a single 5-cycle.
  Edge lengths, in black, are in substitutions per site
  ($\mu(e)g(e)$ in Def.~\ref{def:gene-tree-models}).
  Hybrid edges and their $\gamma$s are shown in blue.}
  \label{fig:5-sunlet-labeled}
\end{figure}

\noindent

Our proof proceeds by establishing two claims: First,
the trees displayed on $N$
have 2 or 3 distinct unrooted topologies, shown in \Cref{fig:disp-trees-top}.
Second, if a network $N$ displays these topologies, then its average distances fit a 5-sunlet.
For this second step, we establish several lemmas.

\begin{figure}[h]
  \centering
  \includegraphics[scale=1]{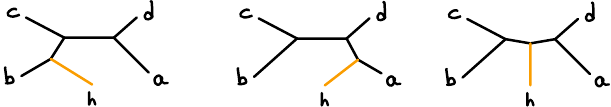}
  \caption{Topologies of trees displayed on $N$ under the conditions in
  \Cref{prop:5-blob-unident-AD}.
  The left and middle topologies must be displayed.
  The right topology may or may not be displayed.}
  \label{fig:disp-trees-top}
\end{figure}

\begin{lemma}\label{prop:5-sunlet-params}
  Let $N_1$ be the 5-sunlet with parameters as shown
  in \Cref{fig:5-sunlet-labeled} (right),
  $\mathfrak{D} = \mathfrak{D}_{N_1}^M$ and $\dsp = \dsp(\mathfrak{D})$,
  where $M$ is the DT model or the NMSCcom model with an expected
  number of substitutions per coalescent unit
  $s = \mu(e)\popsize(e)$ constant across edges $e$ in $N_1$. Then
  \begin{align}\label{eq:5-sunlet-params-1}
    2t_0 &= \dsp_{ac|bd} - \dsp_{ad|bc}
    & 2t_a &= (\dsp_{ab|hd} - \dsp_{ah|bd})/\gamma_b \\ 
    \gamma_b &= (\dsp_{bd|ch} - \dsp_{bc|dh})/(2 t_0) \nonumber 
    & 2t_b &= (\dsp_{ba|hc} - \dsp_{bh|ac})/\gamma_a 
  \end{align}
  and for $j=a,b,c,d,h$, $\mu_j = \hat{\mu}_j$ under the DT model or
  $\mu_j =\hat{\mu}_j - s$ under the NMSCcom model,
  where
  \begin{align}\label{eq:5-sunlet-params-2}
    2\hat{\mu}_c &= \mathfrak{D}(c,b) + \mathfrak{D}(c,d) - \mathfrak{D}(b,d)
    & 2\hat{\mu}_h &= \mathfrak{D}(a,h) + \mathfrak{D}(b,h) - \mathfrak{D}(a,b) \\
    2\hat{\mu}_d &= \mathfrak{D}(d,a) + \mathfrak{D}(d,c) - \mathfrak{D}(a,c) \nonumber
    & 2\hat{\mu}_j &= \mathfrak{D}(j,c) + \mathfrak{D}(j,d) - \mathfrak{D}(c,d) - 2t_j
    \mbox{ \textnormal{for} }j\in\{a,b\}\,.
  \end{align}

  Furthermore, given a metric  $D$ on $\{a,b,c,d,h\}$, there exists a
  5-sunlet $N_1$ as in \Cref{fig:5-sunlet-labeled} (right),
  such that $D=\mathfrak{D}_{N_1}^M$,
  if and only if the following holds, denoting $\dsp=\dsp(D)$:
  \begin{enumerate}
  \setlength{\itemsep}{0pt}
  \item[\textnormal{(a)}] $D$ satisfies the triangle inequality strictly.
  \item[\textnormal{(b)}]
  $D$ satisfies the strict 4-point condition 
  $\dsp_{ad|bc} < \dsp_{ac|bd} = \dsp_{ab|cd}$.
  \item[\textnormal{(c)}]
  On each 4-taxon subset, $D$ is compatible with the circular order
  $(a,d,c,b,h)$:
  \[\begin{split}
  \dsp_{bd|ch} > \max\{\dsp_{bc|dh}, \dsp_{bh|cd}\}, &\quad
  \dsp_{ac|dh} > \max\{\dsp_{ad|ch}, \dsp_{ah|cd}\}  \\
  \dsp_{ab|dh} > \max\{\dsp_{ad|bh}, \dsp_{ah|bd}\}, &\quad
  \dsp_{ab|ch} > \max\{\dsp_{ac|bh}, \dsp_{ah|bc}\}
  \end{split}\]
  \item[\textnormal{(d)}]
  $\hat{\mu}_a > 0$ and $\hat{\mu}_b > 0$ where
  $\hat{\mu}_a$ and $\hat{\mu}_b$ are calculated from $D$ using
  \eqref{eq:5-sunlet-params-1} and \eqref{eq:5-sunlet-params-2}.
  \end{enumerate}
\end{lemma}

\noindent
Note that in each equation of \eqref{eq:5-sunlet-params-1} and
\eqref{eq:5-sunlet-params-2}, the right-hand-side involves
only distances and terms given by preceding equations.
Therefore they allow for
all parameter values to be found from pairwise distances.

\begin{proof}
  By expressing the distances
  $\mathfrak{D}=\mathfrak{D}_{N_1}^\mathrm{DT}$ in terms of parameters
  $t_i, \mu_j, \gamma_b$ where $i = 0, a, b$ and $j = a,b,c,d,h$,
  \eqref{eq:5-sunlet-params-1} and
  \eqref{eq:5-sunlet-params-2} with $\hat{\mu}_j=\mu_j$
  can be verified by calculation.
  For $\gamma_b$, an intermediate but equivalent expression is
  $2\gamma_b - 1 = (\mathfrak{D}(c,h) - \mathfrak{D}(d,h) - \mu_c + \mu_d)/t_0$.

  Under the NMSCcom model,
  $\mathfrak{D}= \sum_T\gamma(T)\; \mathfrak{D}_T^\mathrm{MSC}$ is a
  weighted sum over displayed trees.
  On a tree $T$, assuming a constant mutation rate per coalescent unit
  $s = \mu\popsize$, we get
  $\mathfrak{D}_{T}^\mathrm{MSC}(x,y) = \mathfrak{D}_{T}^\mathrm{DT}(x,y) + 2s$,
  where the extra $2s$ is the average number of substitutions
  after 2 lineages reach a common population (going back in time)
  as it takes an average of 1 coalescent unit for two lineages to coalesce.
  Then $\mathfrak{D}_{N_1}^\mathrm{NMSCcom} = \mathfrak{D}_{N_1}^\mathrm{DT} + 2s$
  and $\dsp^\mathrm{NMSCcom} = \dsp^\mathrm{DT} + 4s$. 
  Therefore \eqref{eq:5-sunlet-params-1} for DT implies
  \eqref{eq:5-sunlet-params-1} for NMSCcom. Also,
  \eqref{eq:5-sunlet-params-2} with $\hat{\mu}_j = \mu_j$ for the DT model implies
  \eqref{eq:5-sunlet-params-2} for the NMSCcom model with $\hat{\mu}_j = \mu_j + s$:
  pendent edge lengths are overestimated by the average number of
  substitutions between speciation and coalescent times.

  For the second part,
  $\mathfrak{D}$ satisfies (a) 
  because pendent edges, being tree edges, are assumed of positive length.
  $\mathfrak{D}$ satisfies (b-c) by
  \Cref{prop:4taxadist} and (d) by the first part.
  Conversely, let $D$ be a metric satisfying (a-d). Then we apply
  \eqref{eq:5-sunlet-params-1} to obtain ``fitted'' parameters
  $\hat{t}_i$ ($i = 0, a, b$) and $\hat{\gamma}=\hat{\gamma}_b$,
  and \eqref{eq:5-sunlet-params-2} to obtain
  $\hat{\mu}_j$ ($j = a,b,c,d,h$).
  We now show that these parameters
  are valid, that is:  $\hat{t}_i > 0$, $0<\hat{\gamma}<1$ and
  $\hat{\mu}_j > 0$.

  We have $\hat{\mu}_a, \hat{\mu}_b > 0$ by (d).
  By (a),
  $\hat{\mu}_h, \hat{\mu}_c, \hat{\mu}_d > 0$.
  From \eqref{eq:5-sunlet-params-1} we get $\hat{t}_0>0$ by (b)
  and $\hat{\gamma} > 0$ by (c).
  Similarly, $\hat{\gamma}_a = 1 - \hat{\gamma} > 0$ and consequently $0<\hat{\gamma}<1$.
  Finally, $t_a,t_b > 0$ by (c).
  $N_1$ can then be assigned these fitted parameters,
  and by the first part, $\mathfrak{D}_{N_1}^\mathrm{DT}=D$.
  Under the NMSCcom model, let
  $s=\min_{j=a,b,c,d,h}\{\hat{\mu}_j/2\}>0$. To edge $e$ in $N_1$
  we assign $\mu(e)=1$ mutation per generation,
  population size $\popsize(e)=1/s$, and length $g(e)$
  generations such that the expected number of substitutions
  $\mu(e)g(e)$ equals $t_i$ ($i=0,a,b$) for internal edges or
  $\hat{\mu}_j-s>0$ ($j=a,b,c,d,h$) for pendent edges.
  Then by the first part, $\mathfrak{D}_{N_1}^\mathrm{NMSCcom}=D$.
\end{proof}

\begin{lemma}
  \label{lem:5-sunlet-displayed-tree-criterion}
  Let $M$ be the DT or NMSCcom model.
  Let $N$ be a binary network on $\{a,b,c,d,h\}$.
  Assume that every tree displayed in $N$, after suppressing degree-2 nodes,
  has one of the topologies shown in \Cref{fig:disp-trees-top},
  with the left and middle topologies both displayed.
  If the pendent edges to $a, b$ in $N$ are sufficiently long,
  then $\mathfrak{D}_{N}^M = \mathfrak{D}_{N_1}^M$
  for a 5-sunlet $N_1$.
\end{lemma}

\begin{proof}
  We will show that $D=\mathfrak D_{N}^M$
  satisfies (a-d) of \Cref{prop:5-sunlet-params}.
  $D$ is a convex combination of $\mathfrak D_{T}^M$
  for trees $T$ displayed in $N$.
  For each displayed tree $T$, $\mathfrak D_{T}^M$ satisfies (a-b)
  so $D$ does too.
  Each expression in (c) is equivalent to two inequalities,
  such as $\dsp_{bd|ch} > \dsp_{bc|dh}$ and $\dsp_{bd|ch} > \dsp_{bh|cd}$.
  Non-strict versions of both inequalities hold for all displayed trees
  by \Cref{prop:4taxadist}, with the strict inequality holding for either the
  left or the middle tree topology. Averaging over displayed trees
  shows the inequalities hold strictly for $D$.
  
  For the final condition (d), we note that when the length of the
  pendent edge to $a$ (resp. $b$) is increased by some amount,
  $\hat{\mu}_a$ (resp. $\hat{\mu}_b$) increases by the same amount,
  and all other fitted parameters are unchanged.
  Hence, for sufficiently long pendent edges to $a$ and $b$, (d) is satisfied.
  The conclusion then follows from \Cref{prop:5-sunlet-params}.
\end{proof}

To prove that the conditions of \Cref{prop:5-blob-unident-AD} imply the
assumptions of \Cref{lem:5-sunlet-displayed-tree-criterion},
a few more definitions will be useful.

\begin{defn}
  Let $N$ be a network on $X$, and $x \in X$. We write
  $N_{\setminus x}$ for the induced subnetwork on ${X \setminus \{x\}}$.
  An \emph{attachment node} of $x$ in $N$
  is a node $u$ in $N_{\setminus x}$ that is, in $N$, incident to an
  edge $e \notin E(N_{\setminus x})$.  The edge $e$ is called an
  \emph{attachment edge}.
\end{defn}

\begin{lemma}
  \label{lem:cutedge-fromdisplayedtrees}
  Let $N$ be a binary network on 5 taxa. If all its unrooted
  displayed trees share the same non-trivial split,
  then $N$ has a cut edge corresponding to that split.
\end{lemma}

\begin{proof}  If $N$ has no $4$- or $5$-blob, the result is trivial, since
  all trees displayed in $N$ have the same unrooted topology as $N$'s tree
  of blobs, whose edges arise from $N$'s cut edges.

  Suppose then that $N$ has a $4$-blob.
  Let $S=a_1a_2|a_3a_4a_5$ be a split in all trees displayed on $N$.
  Since $N$ has a $4$-blob, it has an internal cut edge. If this
  cut edge corresponds to $S$, we are done. Otherwise, this cut edge
is compatible with $S$ because it is present in the displayed trees.
  Without loss of generality, assume that this cut edge corresponds
  to $a_1a_2a_3|a_4a_5$, in which case all trees displayed in $N$ have the
  same unrooted topology (after removing degree-2 nodes) and $N'=N_{\setminus a_5}$ has a 4-blob.
  But by \Cref{lem:4taxonset-displayedsplits} $N'$
  cannot have a 4-blob since all its unrooted displayed trees
  share the same split.

  Now suppose $N$ has a $5$-blob,
  and let $v$ be a lowest hybrid node in the blob.
  Since $N$ is on $5$ taxa, $v$ has exactly one descendent taxon,
  say $a_1$. Assume the split on all displayed trees is
  $S=a_1a_2|a_3a_4a_5$ or $S=a_1a_2a_3|a_4a_5$. In the former case
  we prune $x=a_5$ from $N$ and in the latter case we prune $x=a_3$, to
  consider $N' = N_{\setminus x}$.
  Then $N'$ has a $4$-blob yet all its displayed trees have split
  $a_1a_2|a_3a_4$, which contradicts \Cref{lem:4taxonset-displayedsplits}.
\end{proof}

\begin{proof}[Proof of \Cref{prop:5-blob-unident-AD}]
  Let $N$ be a network satisfying the conditions of \Cref{prop:5-blob-unident-AD}.
  By \Cref{lem:5-sunlet-displayed-tree-criterion}, we just need to show that
  its unrooted displayed trees, after suppressing degree-2 nodes, have one of the
  topologies in \Cref{fig:disp-trees-top}, and that the first two topologies
  are displayed in $N$.
  Let $T$ be a tree displayed in $N$. 
  Then $T_{\setminus h}$ is a tree displayed in $N_{\setminus h}$.
  It must contain cut edges from $N_{\setminus h}$,
  so by assumption 2
  it contains an edge of positive length corresponding to the split $ab|cd$.

  Let $v$ be the attachment node of $h$ in $T$.
  We claim that $v$ is an attachment node of $h$ in $N$ as well.
  Let $h=v_0,v_1,\cdots,v_k=v$ be the path from $h$ to $v$ in $T$,
  and let $v_j$ be the first node along this path to be in $N_{\setminus h}$.
  Then $(v_j,v_{j-1})$ (or $(v_{j-1},v_j)$) is not in  $N_{\setminus h}$
  and $v_j$ is an attachment node of $h$ in $N$. By assumption 3 and since
  $N$ is binary, $v_j$ is incident to 2 cut edges in $N_{\setminus h}$
  and to the attachment edge $(v_j,v_{j-1})$. All 3 edges are then in $T$,
  because $T_{\setminus h}$ must contain all cut edges in $N_{\setminus h}$.
  Therefore $v_j$ has degree 3 in $T$, $j=k$, and $v$ is an attachment node
  of $h$ in $N$.
  Since $v$ is a cut node in $N_{\setminus h}$, and by assumption 3,
  $v$ must be on every up-down path between $a$ and $b$ in $N_{\setminus h}$.
  In particular, $v$ is also on the path between $a$ and $b$ in $T_{\setminus h}$.
  This ensures that $T$ has one of the topologies in \Cref{fig:disp-trees-top}
  (after suppressing degree-2 nodes).

  It remains to show that $N$ displays the left and middle topologies of
  \Cref{fig:disp-trees-top}.
  By assumption 3, only these three trees could be displayed on $N$.
  By assumption 1, $N$ has a $5$-blob, so
  \Cref{lem:cutedge-fromdisplayedtrees}, implies
  the trees displayed on $N$ cannot all share a
  non-trivial split. This rules out only one tree being displayed,
  only the left and right trees (which share $hbc|da$) being displayed, and only
  the middle and right trees (which share $bc|dah$) being displayed.
  Thus either the left and middle trees, or all three trees, are displayed on $N$.
  Applying \Cref{lem:5-sunlet-displayed-tree-criterion} completes the proof.
\end{proof}

\subsection{Limits of identifiability from quartet concordance factors}
 
In this section, we present a family of 5-taxon networks with 5-blobs,
some outer-labeled planar and some not, which are indistinguishable from
a network with a single reticulation
using quartet CFs under the DT, NMSCcom, and NMSCind models.
This family, illustrated in \Cref{fig:CFnetworks}, is a subset of the class considered in
\Cref{prop:5-blob-unident-AD}.

\begin{figure}[h]
  \begin{center}
    \includegraphics[scale=1]{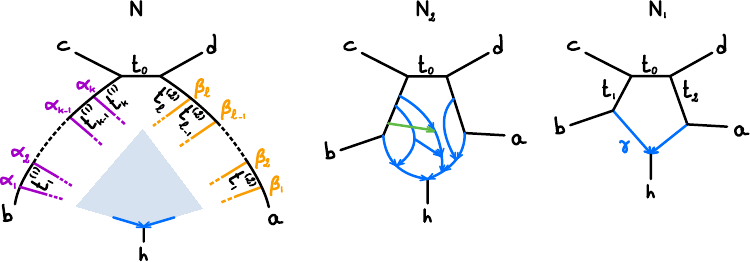}
    \caption{
    Example networks that are not distinguishable from quartet CFs.
    $N$ (left) shows the general structure of a binary network on
    $\{a,b,c,d,h\}$ with $N_{\{a,b,c,d\}}$ the tree $ab|cd$
    with extra degree-2 nodes on the pendent edges to $a,b$.
    At these nodes all paths to $h$ connect.
    $N_1$ (right, outer-labeled planar level-1) and $N_2$ (middle)
    have this structure.
    $N_2$ is not outer-labeled planar,
    but would be without the green edge.
    These networks display the same two tree topologies.
    }\label{fig:CFnetworks}
  \end{center}
\end{figure}

\begin{prop}\label{prop:5-blob-unident-CF}
  Let $M$ be the DT, NMSCcom, or NMSCind model.
  Let $N$ be a metric binary network on $\{a,b,c,d,h\}$ such that:
  \begin{enumerate}
    \setlength{\itemsep}{0pt}
    \setlength{\parskip}{0pt}
    \item The subnetwork $T=N_{\{a,b,c,d\}}$ is a tree with a cut edge that
      induces the split $ad|bc$.
    \item The degree-2 nodes in $T$ at which paths to $h$ last
      leave $N$ all lie on the pendent edges to $a$ or to $b$ in the reduced tree.
  \end{enumerate}
  As labelled in \Cref{fig:CFnetworks} (left),
  let $\alpha_i$ and $\beta_j$ be the probabilities that a lineage from $h$
  traces its ancestry back to specific degree-2 nodes in $T$, and
  let $t_i^{(1)}$ and $t_j^{(2)}$ be the lengths of edges between these nodes.
  Let $N_1$ be the 5-sunlet with topology and parameters as
  depicted in \Cref{fig:CFnetworks} (right).
  If $t_1$, $t_2$, and $\gamma$ are such that
  $\gamma = \sum_{i=1}^k\alpha_i$,
\begin{equation}\label{eq:CFnetworks}
\gamma  e^{-t_1} =   \sum_{i=1}^{k} \alpha_i e^{- \sum_{j=i}^{k} t^{(1)}_j},
\; \mbox{ and } \quad
(1-\gamma)  e^{-t_2} =   \sum_{i=1}^{\ell} \beta_i e^{- \sum_{j=i}^{\ell} t^{(2)}_j}
\end{equation}
  then $N$ and $N_1$ have identical quartet CFs under $M$.
\end{prop}
\begin{proof}
Since $1-\gamma = \sum_{i=1}^\ell\beta_i$,
solving for $t_1$ and $t_2$ necessarily gives $t_1>0$ and $t_2>0$.
Because of the structure of $N$, with at most one lineage passing through any
hybrid node, both the NMSCcom and NSCMind models yield the same formulas.
Let $\boldsymbol{\delta}_j$ denote the standard
basis vectors for $\mathbb{R}^3$, so, for example, $\boldsymbol{\delta}_1=(1,0,0)$. With
$\mathbf{1}= (1,1,1)$ and $y_i=e^{-t_i}$,
we obtain equal quartet CFs for $N$ and $N_1$ given by:
\begin{eqnarray*}
	\cf_{bcda} &=& y_0 \mathbf{1}/3 + (1-y_0) \boldsymbol{\delta}_1 \\
	\cf_{hcda} &=& (\gamma y_0 + (1-\gamma) y_2) \mathbf{1}/3 +
	\gamma (1-y_0) \boldsymbol{\delta}_1 + (1-\gamma) (1-y_2) \boldsymbol{\delta}_3 \\
	\cf_{hbda} &=& (\gamma y_0y_1 + (1-\gamma) y_2) \mathbf{1}/3 +
	\gamma (1-y_0y_1) \boldsymbol{\delta}_1 + (1-\gamma) (1-y_2) \boldsymbol{\delta}_3 \\
	\cf_{hbca} &=& (\gamma y_1 + (1-\gamma) y_0y_2) \mathbf{1}/3 +
	\gamma (1-y_1) \boldsymbol{\delta}_1 + (1-\gamma) (1-y_0y_2) \boldsymbol{\delta}_3 \\
	\cf_{hbcd} &=& (\gamma y_1 + (1-\gamma) y_0) \mathbf{1}/3 +
	\gamma (1-y_1) \boldsymbol{\delta}_1 + (1-\gamma) (1-y_0) \boldsymbol{\delta}_3 \;.
\end{eqnarray*}
Finally, the formulas for CFs under the DT model on the two networks are obtained
by letting all edge lengths go to infinity 
\end{proof}

Note that this construction can be extended to networks obtained by
attaching additional pendent
subnetworks along the edge of length $t_0$ in both $N$ and $N'$.
One can solve for edge lengths $t_1$ and $t_2$ in \eqref{eq:CFnetworks}
simultaneously for all 4-taxon sets.
This yields examples with blobs of arbitrary size.

\subsection{Limits of identifiability from LogDet distances}

In this section, we present a family of 4-taxon rooted ultrametric networks with
rooted 5-blobs, which may or may not be outer-labeled planar,
of arbitrary level, which are indistinguishable
from 4-taxon rooted networks with a single 5-cycle using logDet distances under the
DT, NMSCcom, and NMSCind models with constant mutation rate.

Consider a binary rooted ultrametric network $N^+$ on $\{a,b,c,h\}$, such that the
subnetwork $N^+_{\{a,b,c\}}$ is the rooted tree $a|bc$ with one binary node
introduced on each of the pendent edges to $a,b$. At these binary nodes all
paths to $h$ leave $N^+_{\{a,b,c\}}$, though the subgraph of edges ancestral only
to $h$ is otherwise unrestricted.
\Cref{fig::logdetnetworks}  illustrates
the general structure ($N^+$, left) and  gives examples ($N^+_1$ and $N^+_2$).
Let $\alpha$ and $\beta=1-\alpha$ be the probabilities that a lineage from
$h$ traces its ancestry back to the nodes indicated in the figure,
determined by the hybrid parameters for edges ancestral only to $h$.  

\begin{figure}[h]
  \begin{center}
    \includegraphics[scale=1]{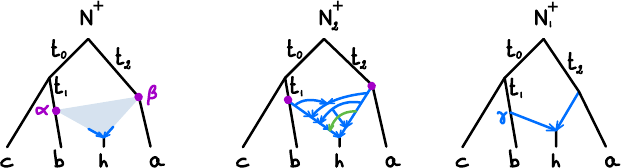}
    \caption{
    Example rooted networks that are not distinguishable from pairwise logDet distances. 
    $N^+$ (left) shows the general structure of an ultrametric binary network on
    $\{a,b,c,h\}$ such that $N^+_{\{a,b,c\}}$ is
    the rooted tree $ab|c$, incident to the funnel of $h$ only once on its
    pendent edges to $a$ and to $b$.
    $N^+_1$ (right, outer-labeled planar utrametric level-1) and
    $N^+_2$ (middle) have this structure.
    $N^+_2$ is not outer-labeled planar, but would be if
   the green edge were removed.
    These networks have the same two rooted displayed tree topologies.
    }\label{fig::logdetnetworks}
  \end{center}
\end{figure}

For such a network $N^+$,
consider the rooted ultrametric level-1 network $N^+_1$ in \Cref{fig::logdetnetworks} (right)
where $\gamma = \alpha$.
Since no coalescent event can occur in the funnel of $h$ on $N^+$ or $N^+_1$,
these networks yield the same metric gene tree distribution and therefore same
pairwise logDet distances under the DT, NMSCcom, and NMSCind models with constant mutation rate.
Analogously to the example given for the limits of identifiability under CFs,
these examples generalize to networks with larger blobs.

\medskip

In fact under the standard  models we consider, the networks of this section are
nonidentifiable from \emph{any} sequence data if only one individual is sampled per taxon.
Indeed, from such sequences, we can at best identify metric gene trees,
but these do not distinguish the networks \cite{2015PardiScornavacca}.
While data from multiple samples from the hybrid taxon $h$ may give additional
information, it remains unclear what structure within the blob may
then be identified
(but see \cite{2017ZhuDegnan} for an example of two non-outer-labeled
planar rooted networks with identical displayed trees,
distinguishable using 2 samples from the hybrid taxon).

\section*{Acknowledgements}
We thank Elizabeth Allman for insightful discussions, advice and feedback.
This work was supported in part by the National Science Foundation through grants
DMS 2023239 to C.A., DMS 2331660 to H.B., and DMS 2051760 to J.R., and by a H. I. Romnes faculty fellowship
to C.A. provided by the University of Wisconsin-Madison Office of the
Vice Chancellor for Research and Graduate Education with funding from the
Wisconsin Alumni Research Foundation.

\appendix
\section{Extension to non-binary networks}

\renewcommand{\thesection}{\Alph{section}} 

We now prove an extension of \Cref{lem:4taxonset-displayedsplits}.
For binary networks, there are only 2 cases in
\Cref{lem:4taxonset-displayedsplits}:
a network has either a cut edge separating $ab$ from $cd$ up to relabeling
of leaves (Case 1) or has a 4-blob (Case 2).
When relaxing the binary assumption,
Case 1 extends to networks with a blob that admits $ab$ and/or $cd$ as a taxon block.
Case 2 remains unchanged, but the network may also display the unresolved star tree.
All other networks (Case 3) display the star tree only.

\begin{lemma}\label{lem:4taxonset-displayedsplits-non-binary}
  Let $N$ be a network on $\{a,b,c,d\}$.
  Then one of the following holds:
\begin{enumerate}
\item
  $N$ has a 2- or 3-blob with at least one 2-taxon block, say $ab$,
  and displays the quartet $ab|cd$ and no other resolved quartet.
  If $N$ is also outer-labeled planar, it is congruent with
  circular orders $(a,b,c,d)$ and $(a,b,d,c)$.
\item
  $N$ has a 4-blob, and 
  displays at least 2 resolved quartets.
  If the 4-blob of $N$ is also outer-labeled planar, with
  unique circular order $(a,b,c,d)$, then $N$
  displays exactly 2 resolved quartets,
  $ab|cd$ and  $ad|bc$.
\item
 
  $N$ has a \emph{central node} whose deletion leaves the taxa
  in 4 distinct connected components, and
  displays only the unresolved star quartet.
\end{enumerate}
\end{lemma}

Simple examples show that $N$ may or may not also display
the unresolved star quartet if it has no 2-blob with a 2-taxon group,
regardless of whether it has a 3- or 4-blob.

In the proof below, we call a network's blob \emph{pendent} to taxon $x$ if
it is a 2-blob with taxon blocks $x$ versus all the other taxa.
The \emph{pendent blob chain} to taxon $x$ is the subnetwork consisting
of all nodes and edges in blobs pendent to $x$. 

\begin{proof}
 We first prove that the blob conditions of Cases 1, 2 and 3 are exhaustive.
   On 4 taxa any 3-blob must have a taxon block of size 2,
 so if the conditions of Case~1 and 2 are not met,
  then $N$ does not have any 3-blob, 4-blob, or 2-blob with a taxon block of size 2.
  Then we claim that $N$ has a central node $z$. 
 
  To see this, consider the block-cut tree $G$ of $N$, with node set comprised
  of $N$'s blobs and cut nodes, and edges joining
  blobs and cut nodes adjacent in $N$.
  $G$ is a tree \cite[Lemma 3.1.4]{Diestel2017}.
  Pendent cut edges in $N$ form the leaves of $G$.
  An $m$-blob in $N$ corresponds to a degree-$m$ vertex in $G$.
  In the absence of degree-3 blobs in $N$,
  a degree-3 vertex in $G$ corresponds to a cut node in $N$
  incident to three 2-blobs, one of which must have a taxon block of size 2.
  Therefore  
  $G$ has no degree-3 vertex, and has a degree-4 vertex
  corresponding to a central cut node $z$ in $N$.
 Then $N$ consists of 4 pendent blob chains adjacent to $z$,
  one to each taxon; and all displayed trees are stars with $z$ as
  a central node. 
 If $N$ is outer-labeled planar,
  that is, if its pendent chains are planar,
  then $N$ is easily seen to be congruent with all 3 circular orders. 

  \smallskip
  In Case 1, assume that $N$ has a blob $B$ with $ab$ as one
  taxon block.
  Within a tree $T$ displayed in $N$, the edge(s) retained from
  $B$ must then map to edge(s) in $T$, one of which must be a cut edge separating
  $ab$ from the other 2 taxa, as claimed.
  To show that such an outer-labeled planar $N$ is congruent with
  the orders $(a,b,c,d)$ and $(a,b,d,c)$,
  we build outer-labeled planar networks displaying $N$,
  each with a 4-blob and these orders.
  Since a pendent  edge  can be directed towards the leaf,
  we  build these from $N$ by adding a hybrid edge
  from either of the edges incident to $c$ or $d$ to the edge incident to $a$,
  after rotating part of the planar network about the blob's articulation node
  for $ab$ if necessary.

  \smallskip
  For Case 2, consider $N$ with a 4-blob $B$.
  Let $v$ be a lowest hybrid node in $B$.
  We may assume that removing any parent edge of $v$
  (and then edges no longer on up-down paths between taxa) gives a subnetwork
  \emph{without} a 4-blob. Otherwise, we remove such edges iteratively
  until we cannot do so while retaining a 4-blob, and observe that any quartet displayed
  in the modified network is also displayed in $N$.
  Let $d$ be the taxon below $v$;
  $e_1=(v_1,v),\ldots,e_k=(v_k,v)$ the parent edges of $v$; and 
  $N_j$ the network obtained from $N$
  by removing $e_j$ and nodes and edges no longer an up-down paths between taxa, 
  so as to prune unlabeled leaves. We will prove that at least two of
  $N_1,\ldots,N_k$ fall under Case 1 for distinct quartets.
  For the sake of contradiction,
  assume that for all $j$, $N_j$ falls under Case 1 for quartet $ab|cd$
  (and so has a blob with taxon group $ab$ or $cd$)
  or under Case 3 (and has a central node whose deletion separates all taxa).

  Consider the subnetwork $\widetilde{N}$ induced by $a,b,c$.
  Then $\widetilde{N}$ is a subgraph of $N_j$ for all $j$.
  If $\widetilde{N}$ has a central node whose deletion separates
  all 3 taxa $a,b,c$, call this node $z$. Otherwise,
  $\widetilde{N}$ must have a 3-blob
(corresponding to the degree-3 node in $\widetilde{N}$'s block-cut tree),
  and let $z$ be its articulation node leading to $c$.
  In both cases, $z$ belongs to the pendent blob chain to $c$ in
  $\widetilde{N}$, and separates this blob chain
  from the subnetwork of $\widetilde{N}$ (and of $N$)
  induced by $a,b$.

  Since a blob $\widetilde{B}$ of $\widetilde{N}$ is
  biconnected, any blob $B$ of $N_j$ either contains
  $\widetilde{B}$, or shares no edge with $\widetilde{B}$.
  If $N_j$ has a blob $B$ with taxon block $ab$, then
  $B$ may only contain blobs from $\widetilde{N}$ that are
  pendent to $c$, therefore $N_j$ is formed from $\widetilde{N}$
  by attaching a subgraph to the pendent blob chain to $c$.
  If $N_j$ has no blob with taxon block $ab$
  but a blob $B$ with taxon block $cd$,
  then $B$ is on $\widetilde{N}$, 
  and $N_j$ is formed from $\widetilde{N}$ by attaching a
  pendent blob chain to $d$ at $z$.

  Thus in either case, $N_j$ results from attaching a subgraph
  to $\widetilde{N}$
  only along the pendent blob chain to $c$. Therefore $z$ disconnects
  the subnetwork induced by $a,b$ from the subnetwork induced by $c,d$
  in the full network $N$,
  and $N$ cannot have a 4-blob --- a contradiction.

  The first claim of Case 2 follows easily, having established Case 1
  and that at least two of $N_1,\ldots,N_k$ fall under Case 1
  for distinct quartets.
  If $N$ is outer-labeled planar,
  by Corollary \ref{cor:uniquecirc} it has a unique circular order, say $(a,b,c,d)$.
   Then by \Cref{lem:circ4taxontree},
  the displayed quartets can only be $ab|cd$ and $ad|bc$.
\end{proof}

\bibliographystyle{alpha}
\bibliography{Hybridizationv2}

\end{document}